\documentclass[journal]{IEEEtran}

\usepackage{cite}
\usepackage{amsmath,amssymb,amsfonts}
\usepackage{amsthm}
\usepackage[linesnumbered,ruled,vlined]{algorithm2e}
\usepackage{enumitem}
\usepackage{graphicx}
\usepackage{textcomp}
\usepackage{dsfont}
\usepackage{xcolor}
\usepackage{bm}
\usepackage{soul}
\def\BibTeX{{\rm B\kern-.05em{\sc i\kern-.025em b}\kern-.08em
    T\kern-.1667em\lower.7ex\hbox{E}\kern-.125emX}}

\usepackage{changepage}
\usepackage{etoolbox}
\usepackage{subfigure}
\usepackage{ifthen}
\newboolean{showcomments}
\setboolean{showcomments}{true}

\newcommand{\Meng}[1]{\ifthenelse{\boolean{showcomments}}
{ \textcolor{blue}{(Meng says:  #1)}}{}}

\newcommand{\RY}[1]{\ifthenelse{\boolean{showcomments}}
{ \textcolor{blue}{(Roy says:  #1)}}{}}

\newcommand{\MQ}[1]{\ifthenelse{\boolean{showcomments}}
{ \textcolor{blue}{(Mengqiu:  #1)}}{}}

\usepackage{mathtools}
\SetKwComment{Comment}{$//$}{}

\newcommand{\bs}{\boldsymbol}

\DeclarePairedDelimiter{\set}{\{}{\}}

\DeclareMathOperator*{\argmin}{argmin}

\newcommand{\nn}{\nonumber\\}

\newtheorem{theorem}{Theorem}

\newtheorem{remark}{Remark}
\newtheorem{proposition}{Proposition}
\newtheorem{lemma}{Lemma}

\newtheorem{definition}{Definition}
\newtheorem{assumption}{Assumption}

\newtheorem*{question}{Key Question}
\newtheorem{corollary}{Corollary}

\begin{document}
\title{
Timely Best Arm Identification in Restless Shared Networks
}

\author{Mengqiu~Zhou,~\IEEEmembership{Student Member,~IEEE,}
Vincent~Y.~F.~Tan,~\IEEEmembership{Senior Member,~IEEE,}
and~Meng~Zhang,~\IEEEmembership{Member,~IEEE}
\thanks{Mengqiu Zhou is with the College of Information Science and Electronic Engineering and ZJU-UIUC Institute, Zhejiang University, China (e-mail: mengqiuzhou@zju.edu.cn).}
\thanks{Vincent Y. F. Tan is with the Department of Mathematics, the Department of Electrical and Computer
Engineering, and the Institute of Operations Research and Analytics, National University of Singapore (e-mail: vtan@nus.edu.sg).}
\thanks{Meng Zhang is with the ZJU-UIUC Institute, Zhejiang University, China (e-mail: mengzhang@intl.zju.edu.cn).}
}

\markboth{Journal of \LaTeX\ Class Files,~Vol.~14, No.~8, August~2021}%
{Shell \MakeLowercase{\textit{et al.}}: A Sample Article Using IEEEtran.cls for IEEE Journals}


\maketitle

\begin{abstract}

    Real-time status updating applications increasingly rely on networks of devices and edge nodes to maintain data freshness, as quantified by the age of information (AoI) metric. Given that edge computing nodes exhibit uncertain and time-varying dynamics, it is essential to identify the optimal edge node with high confidence and sample efficiency, even without prior knowledge of these dynamics, to ensure timely updates.
    To address this challenge, we introduce the first best arm identification (BAI) problem aimed at minimizing the long-term average AoI under a fixed confidence setting, framed within the context of a restless multi-armed bandit (RMAB) model. In this model, each arm evolves independently according to an unknown Markov chain over time, regardless of whether it is selected.
    To capture the temporal trajectories of AoI in the presence of unknown restless dynamics, we develop an age-aware LUCB algorithm that incorporates Markovian sampling. Additionally, we establish an instance-dependent upper bound on the sample complexity, which captures the difficulty of the problem as a function of the underlying Markov mixing behavior.
    Moreover, we derive an information-theoretic lower bound to characterize the fundamental challenges of the problem. We show that the sample complexity is influenced by the temporal correlation of the Markov dynamics, aligning with the intuition offered by the upper bound.
    Our numerical results show that, compared to existing benchmarks, the proposed scheme significantly reduces sampling costs, particularly under more stringent confidence levels.
\end{abstract}

\section{Introduction}
\subsection{Background and Motivations}
Information freshness is becoming increasingly significant due to the rapid growth of real-time applications.
For instance, in vehicular networks, timely status updates regarding traffic conditions and vehicle status are vital for reliable autonomous driving~\cite{liu2020computing}.
Similarly, industrial control systems depend on immediate alerts of equipment anomalies, such as overheating and abnormal pressure levels, to maintain operational safety~\cite{moyne2007emergence}.
In addition, real-time health monitoring updates from wearable devices enables prompt medical interventions~\cite{kakria2015real}.
To quantify data freshness, the \textit{age of information (AoI)} metric has been introduced \cite{yates2021age}, which measures the time that has elapsed since the most recent data update.

Fresh data delivery in such applications critically relies on Mobile Edge Computing (MEC) infrastructures~\cite{ren2019survey}, where computational tasks are offloaded to edge nodes such as base stations or roadside units located at the network edge.
A fundamental challenge in public edge networks lies in the uncertainty of shared resources. 
Unlike private servers, edge nodes operate as shared infrastructure serving a large number of background users.
In practice, users are typically served through virtualized compute slices with standardized resource specifications at each edge node, while performance variability primarily arises from time-varying background workloads.

From the perspective of a specific user (e.g., a connected vehicle), the internal congestion state of an edge node is unobservable and evolves continuously over time, regardless of whether the user selects it. 
As a result, a node that delivered fresh updates moments ago may become heavily congested due to sudden bursts of background traffic.
This creates a \emph{restless} environment in which the user must repeatedly decide which edge node to query without direct knowledge of its instantaneous load.

Extensive research has investigated AoI optimization focusing on task offloading strategy~\cite{chen2021information,chen2023joint}, transmission scheduling~\cite{kadota2018scheduling,hsu2018age,saurav2021minimizing} and resource allocation mechanism~\cite{yates2018age,chen2020age,tripathi2024whittle}.
Regrettably, all of these efforts rely on prior knowledge of system dynamics.
In particular, while recent pioneering studies have formulated AoI minimization problems within the restless multi-armed bandit (RMAB) framework and designed index-based scheduling policies~\cite{hsu2018age,tripathi2024whittle}, they assume known transition probabilities to enable tractable solutions such as Whittle’s index~\cite{whittle1988restless}.

In practical open edge networks, continuously learning or tracking the complete, time-varying transition models of all edge nodes is often computationally prohibitive and resource-intensive, as such information is typically hidden from individual users.
Moreover, although dynamic scheduling could theoretically exploit instantaneous idle slots, it is frequently impractical in stateful edge applications due to prohibitive switching overheads, including context migration and connection re-establishment costs.
Therefore, the crucial requirement in practice is to quickly and reliably identify a stable and high-performing edge node that ensures data freshness under unknown system dynamics, rather than continuously switching among nodes.
This motivates our key question:
\begin{question}
    How should one design a sample-efficient age-optimal best edge node identification strategy under unknown and time-varying dynamics?
\end{question}

To address the unobservability of network state, we model the internal dynamics of each edge node using a Hidden Markov Model (HMM), where the hidden state represents a congestion level induced by background traffic.
Unlike classical queueing models that assume observable queue lengths~\cite{kleinrock1974queueing}, our abstraction captures the limited visibility faced by users in realistic edge networks.
We formulate this problem as an age-optimal Best Arm Identification (BAI) task under a fixed-confidence setting within the RMAB framework, aiming to identify the edge node that minimizes the long-term average AoI with minimal sampling cost.


Solving this problem is particularly challenging due to the presence of unknown restless dynamics and temporal correlations in observations.
Classical concentration inequalities designed for i.i.d. samples (e.g., Chernoff \cite{chernoff1952measure} or Hoeffding \cite{hoeffding1963probability} bounds) are no longer applicable, and the autonomous evolution of unselected nodes further complicates learning under partial observations.

\subsection{Contributions}
Our key contributions are summarized as follows.
\begin{itemize}[leftmargin=*]
    \item \textit{Problem Formulation}: We propose a novel age-optimal best arm identification problem in the restless multi-armed bandit framework under a fixed-confidence setting.
    To the best of our knowledge, this is the first work to identify edge node performance under unknown and restless dynamics, with the objective of minimizing the average age of information.
    
    \item \textit{Age-Optimal Best Arm Identification}: 
    To overcome the nonconvexity of our age-optimal BAI problem which has a fractional cost, we transform it into an equivalent average cost problem.
    To handle the partial observability in RMAB, we develop an age-aware LUCB algorithm which incorporates a Markovian sampling strategy. 
    Further, we establish an instance-dependent upper bound on the sample complexity and show that the bound depends on the mixing behavior of underlying Markovian dynamics.

    \item \textit{Lower Bound:} 
    We derive a family of information-theoretic lower bounds on the sample complexity
    parameterized by a certain quantity known as the delay constraint $D$.
    We conduct a case study for a two-arm, two-state case with $D=2$. This study reveals that the fundamental sample complexity is critically shaped by the temporal correlation of the Markov dynamics, consistent with the mixing dependence observed in the upper bound analysis.
    
    \item \textit{Numerical Results}: We compare the age-optimal BAI scheme to three benchmarks in terms of sample complexity. Numerical results show that our scheme can reduce the sample complexity by up to $43\%$ relative to the benchmarks, with greater savings achieved under stricter confidence requirements.
    Moreover, the sample complexity gains become increasingly pronounced for harder instances characterized by smaller AoI gaps.
\end{itemize}

\section{Related Work}
In this section, we briefly review the literature related to our proposed age-optimal BAI. 
Related works can be classified into two categories: AoI and RMAB.
\subsection{Age of Information}
AoI, which measures the freshness of status update, has motivated extensive research since its introduction in \cite{kaul2012real}. 
Existing works along this line mainly focus on minimizing time-average AoI under a variety of system settings, including queueing networks (e.g., \cite{bedewy2019minimizing,talak2020age,bedewy2019age,yates2018age}), computing systems (e.g., \cite{arafa2019timely,zhou2024age}) and wireless networks (e.g., \cite{arafa2019age,sun2019sampling,sun2017update}).
Some pioneering works have recently adopted the restless multi-armed bandit (RMAB) framework to address AoI-aware source scheduling problems \cite{kadota2018optimizing,hsu2018age,tripathi2024whittle}, in which Whittle’s index is used to design low-complexity policies.
These approaches typically assume full knowledge of the underlying Markov transition model and rely on structural properties such as indexability to derive closed-form indices.
\textit{However, none of these works consider the age minimization problem under unknown system dynamics due to the unobservable congestion status, where such index-based methods are not directly applicable.}
To fill this gap, we study the best edge node identification problem under unknown Markovian dynamics, aiming to minimize the average AoI in this paper.

\subsection{Restless Multi-Armed Bandit}
The RMAB generalizes the classical bandit by allowing each arm to evolve continuously over time, 
regardless of selection.
Since its introduction by Whittle in \cite{whittle1988restless}, the RMAB has been extensively studied due to its modeling flexibility and analytical challenges.
Prior research on RMAB problems primarily focused on minimizing cumulative regret over a fixed horizon, has explored a variety of solution approaches, including index-based heuristic policies (e.g., \cite{guha2010approximation,liu2010indexability,akbarzadeh2022conditions}), 
regret-based online learning approaches (e.g., \cite{tekin2012online,ortner2012regret,jiang2023online,wang2023optimistic}), 
and reinforcement learning (e.g., \cite{wang2020restless,nakhleh2021neurwin}).
However, in mission-critical applications (e.g., autonomous driving) where data staleness can compromise safety, cumulative regret minimization is often insufficient.
In such scenarios, the BAI problem
which aims to identify the best arm with high confidence using as few samples as possible is more suitable, yet receives limited attention in the RMAB context.
To the best of our knowledge, only \cite{Karthik2023best} has addressed this issue.
\textit{Nevertheless, all of these efforts define the best arm as the one with the highest expected reward, without considering structured objectives such as Age of Information, which poses unique challenges as it requires tracking over temporal trajectories rather than instantaneous rewards.}

\section{System Model}
\label{sec:system model}
In this section, we consider a status update system illustrated in Fig.~\ref{fig:model}, consisting of a source, a scheduler, $K$ edge nodes with unknown congestion dynamics induced by background traffic, and a monitor.
The scheduler submits update packets to one of the edge nodes, each of which operates under a first-come-first-served (FCFS) discipline and is subject to additional, unobservable workloads generated by other users.
We now start with models for the restless edge node model and the time-average age metric in order to then formulate the age-optimal best arm identification problem via restless multi-armed bandit framework.
For any positive integer $A$, we use $[A]\triangleq\{1,2,\ldots,A\}$ to denote the set of integers up to $A$.
\begin{figure}[t]
    \centering
    \includegraphics[width=\linewidth]{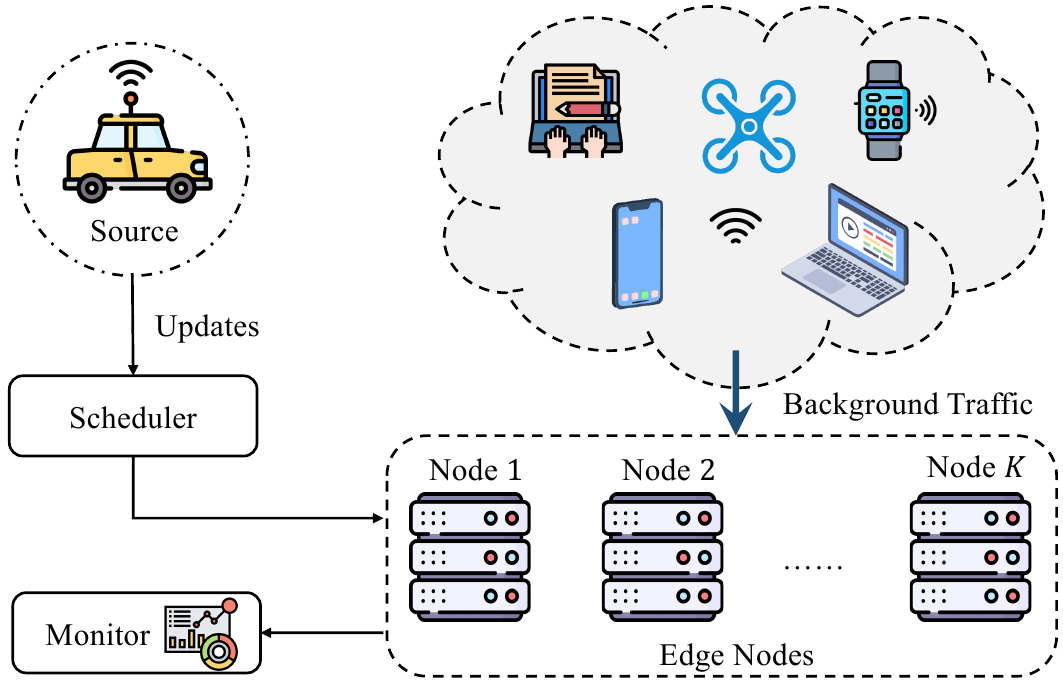}
    \caption{The system model. Each edge node is modeled as an arm in a restless multi-armed bandit framework, with unknown congestion dynamics.}
    \label{fig:model}
\end{figure}

\subsection{Restless Edge Node Model}\label{sec:restless_server}
We consider a cluster of $K$ edge nodes, each of which is shared by multiple services and is responsible for processing update packets from external sources in the network.
Suppose the $i$th update packet is submitted to edge node $a\in[K]$.
At that instant, the node is characterized by a congestion state
$X_{a,i}\in\mathcal{S}=\{0,1,\ldots,S\}$, which represents the background workload and contention effects induced by other users in the network.

We model the background traffic dynamics of each node as a stationary ergodic Markov chain with unknown transition probability denoted by $P_a(s'|s)=\mathrm{Pr}[X_{a,i+1}=s'|X_{a,i}=s]$.
Due to the lack of visibility into the internal congestion levels in practice, the scheduler does not observe $X_{a,i}$ directly and instead observes the service delay $Y_{a,i}\in\mathcal{Y}$ upon completion of update $i$ processed by edge node $a$.
As a result, the pair $(X_{a,i},Y_{a,i})$ forms a Hidden Markov Model (HMM), where the emission probability is defined as $\mathrm{Pr}\left[Y_{a,i}=y|X_{a,i}=s\right]$.

Motivated by network slicing techniques~\cite{9355609,10173679} widely applied in practical mobile edge computing systems, the user is served by a standardized compute slice (e.g., a container with fixed CPU/bandwidth quotas) at each node.
In such virtualized environments, the service delay experienced by the user is primarily governed by background congestion.
Accordingly, we start with adopting a deterministic emission model as a first-order approximation, where $\Pr[Y_{a,i}=y\mid X_{a,i}=s]=1$ and the observation space is $\mathcal{Y}=\{f(0),f(1),\ldots,f(S)\}$.

The congestion state $X_{a,i}$ is driven by exogenous background traffic.
In a spatially distributed edge network, traffic patterns across different locations (e.g., business districts versus residential areas) are typically weakly correlated. We therefore assume that the state transitions of different edge nodes are independent.
Under this setting, the transition probability of each arm varies primarily
in its background arrival characteristics, which are commonly modeled as Poisson
processes.
To capture a broad class of transition behaviors while retaining analytical tractability, we assume that each node's transition probability matrix is generated from a one-parameter exponential family studied in \cite{moulos2019optimal}.

We introduce a real-valued parameter $\theta_a\in\mathbb{R}$ to characterize the transition dynamics of each arm $a$ and denote its transition probability matrix by $P_{\theta_a}$.
Specifically, we fix an irreducible base transition matrix $P$ over $\mathcal{S}$ and construct the unnormalized matrix for any edge node $a$ given as
\begin{align}
    \tilde{P}_{\theta_a}(s'|s)\triangleq P(s'|s)\exp(\theta_a\cdot f(s')), \quad\forall s,s'\in\mathcal{S}. \label{eq:tilde_P}
\end{align}
We denote the full collection of arm-specific parameters by $\bs{\theta}\triangleq[\theta_1,\ldots,\theta_K]^{\top}$ and refer to $\bs{\theta}$ as the problem instance.

Note that $\tilde{P}_{\theta_a}$ defined in \eqref{eq:tilde_P} is not a valid stochastic transition probability matrix as its rows do not necessarily sum up to one.
To address this, we normalize \eqref{eq:tilde_P} by invoking Perron-Frobenius theory.
For each $\theta_a$, let $\rho(\theta_a)$ denote the Perron–Frobenius eigenvalue of $\tilde{P}_{\theta_a}$ and $\bs{v}_{\theta_a}=[v_{\theta_a}(s)\colon s\in\mathcal{S}]^\top$ be its corresponding unique positive right eigenvector~\cite[Theorem 8.4.4]{horn2012matrix}.
The normalized transition probability matrix $P_{\theta_a}$ is specified by
\begin{align}\label{eq:P_theta}
    P_{\theta_a}(s'|s)=\frac{v_{\theta_a}(s')}{\rho(\theta_a)v_{\theta_a}(s)}\tilde{P}_{\theta_a}(s'|s), \quad s,s'\in\mathcal{S}.
\end{align}

To ensure that each $P_{\theta_a}$ defines as an ergodic Markov chain, we begin by defining $\bar{f}=\max\limits_{s} f(s)$ and $\underline{f}=\min\limits_{s} f(s)$, and the associated level sets $\mathcal{S}_{\bar{f}}=\{s\in \mathcal{S}\colon f(s)=\bar{f}\}$ and $\mathcal{S}_{\underline{f}}=\{s\in \mathcal{S}\colon f(s)=\underline{f}\}$.
We then impose the following mild conditions on the stochastic matrix $P$.
\begin{assumption}\label{assump:P}
    We assume that the matrix $P$ satisfies:
    \begin{itemize}
        \item[\textbf{A1:}] The submatrix of $P$ with rows and columns in $\mathcal{S}_{\bar{f}}$ is irreducible. 
        \item[\textbf{A2:}] For every $s\in \mathcal{S}\setminus\mathcal{S}_{\bar{f}}$, there exists $s'\in\mathcal{S}_{\bar{f}}$ such that $P(s'|s)>0$. 
        \item[\textbf{A3:}] The submatrix of $P$ with rows and columns in $\mathcal{S}_{\underline{f}}$ represents an irreducible Markov chain. 
        \item[\textbf{A4:}] For every $s\in \mathcal{S}\setminus\mathcal{S}_{\underline{f}}$, there exists $s'\in\mathcal{S}_{\underline{f}}$ such that $P(s'|s)>0$. 
    \end{itemize}
\end{assumption}
Assumption \ref{assump:P} accommodates a wide range of models as it only requires partial irreducibility and accessibility to irreducible subsets.
In particular, any transition matrix $P$ with strictly positive entries trivially satisfies the above conditions.
Based on \cite{meyn2012markov}, we establish the following lemma.
\begin{lemma}\label{lemma:P_theta}
    When $P$ represent an irreducible transition probability matrix on the finite state space $\mathcal{S}$ satisfying Assumption~\ref{assump:P}, for any $\theta_a$, the matrix $P_{\theta_a}$ is irreducible and positive recurrent. 
\end{lemma}
Lemma~\ref{lemma:P_theta} implies that $P_{\theta_a}$ has a unique stationary distribution, which we denote by $\mu_{\theta_a}=[\mu_{\theta_a}(s)\colon s\in\mathcal{S}]^\top$.
Since the underlying congestion state of each node continues to evolve even when not selected, we adopt a RMAB framework by modeling each node as an independent arm whose state evolves over a finite state space $\mathcal{S}$ according to a discrete-time, time-homogeneous Markov process.

\subsection{Time-Average Age}
An \textit{age of information (AoI)} metric is used to characterize timeliness by the vector of ages tracked by monitors~\cite{yates2021age}.
As depicted in Fig.~\ref{fig:model}, the updating system has a source that submits updates with computational requirements to a scheduler, which selects one of the $K$ edge nodes for processing before delivering those updates to a destination monitor.

Suppose the source observes an acknowledgement from the monitor and generate fresh updates as a stochastic process at times $\set{t_i \colon i\in\mathbb{N}}$, with each update time-stamped $t_i$.
We first consider the zero-waiting policy.\footnote{Optimizing the waiting strategy in age-optimal BAI problem will be further investigated in our future work.}
These updates are processed by the node and delivered to the destination monitor at times $\set{t'_i \colon i\in\mathbb{N}}$.
This induces the AoI process $\Delta(t)$ has a characteristic sawtooth shape depicted in Fig.~\ref{fig:age}.
Denote the service time for update $i$ as $Y_{i}$, the AoI process $\Delta(t)$ at the destination monitor is reset at time $t'_i$ to $\Delta(t'_i)=Y_{i}$, the age of the update received at that time instant.

\begin{figure}[t]
    \centering
    \includegraphics[width=\linewidth]{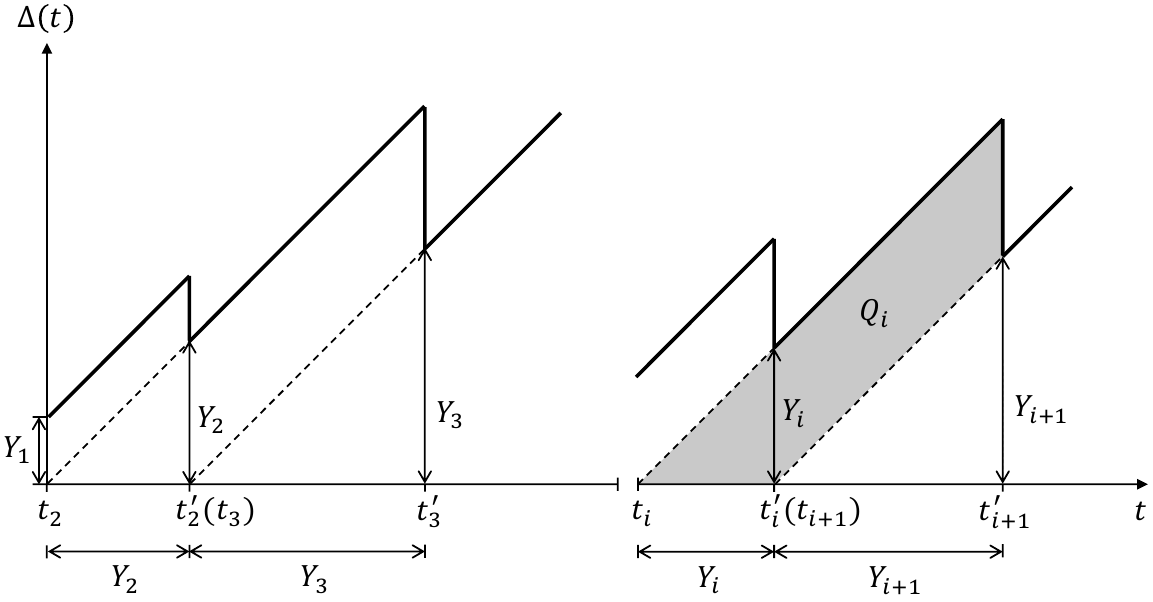}
    \caption{Age of information $\Delta(t)$ evolution in time under zero-wait policy.}
    \label{fig:age}
\end{figure}

AoI analysis was initiated in \cite{kaul2012real}, which defines the time-average AoI as
\begin{align}
    {\Delta}^{(\text{ave})}=\lim_{T\to\infty} \frac{1}{T}\int_0^T\Delta(t)\,dt.
\end{align}
This time-average was evaluated by graphically decomposing the integral into areas $Q_i$, as shown in Fig.~\ref{fig:age}. 
Following this approach, we define epoch $i$ as the time interval $[t_{i},t'_{i})$ and we observe from Fig.~\ref{fig:age} that associated with this epoch is the area 
\begin{IEEEeqnarray}{rCl} 
    Q_i=Q(Y_i,Y_{i+1})&\triangleq&
    \frac{1}{2}\left(Y_i+Y_{i+1}\right)^{2}
    -\frac{1}{2}Y_{i+1}^{2}\nn
    &=& \frac{1}{2}Y_i^2+Y_i Y_{i+1}.\IEEEeqnarraynumspace\label{eq:Q-defn} 
\end{IEEEeqnarray}
It follows that the long-term average AoI is thus given by \cite{yates2021age}:
\begin{align}
    {\Delta}^{(\text{ave})}\triangleq&\limsup_{n\rightarrow \infty}\frac{ \sum_{i=1}^n \mathbb{E}[Q(Y_i,Y_{i+1})]}{\sum_{i=1}^n \mathbb{E}[Y_i]}. \label{average-AoI}
\end{align}

\subsection{Problem Formulation}
Our goal is to minimize the long-term average AoI ${\Delta}^{(\mathrm{ave})}$ by selecting the best node $a\in[K]$ with high confidence and minimal samples. 
We use $Y_{a,i}$ to denote the service time experienced by the $i$th packet processed by node $a\in[K]$, then the problem can be formulated as 
\begin{align}\label{eq:problem_aoi}
    \min_{a\in[K]}\;\limsup_{n\rightarrow \infty}\frac{\sum_{i=1}^n \mathbb{E}[Q(Y_{a,i},Y_{a,i+1})]}{\sum_{i=1}^n \mathbb{E}[Y_{a,i}]}.
\end{align}
It follows from Lemma \ref{lemma:P_theta} that each node $a$ is associated with a Markov chain with a unique stationary distribution. Hence, we drop the index $n$ and Problem \eqref{eq:problem_aoi} can be simplified as
\begin{align}\label{eq:stationary_problem}
    \gamma^*\triangleq\min_{a\in[K]}\;\frac{\mathbb{E}[Q(Y_{a,i},Y_{a,i+1})]}{\mathbb{E}[Y_{a,i}]}.
\end{align}
In order to overcome the nonconvexity of the objective $\mathbb{E}[Q(Y_{a,i},Y_{a,i+1})]/\mathbb{E}[Y_{a,i}]$, we use the fractional programming technique by optimizing the expected value of the difference 
\begin{equation}\label{eq:Gdefn}
    D(Y_{a,i},Y_{a,i+1},\gamma) \triangleq Q(Y_{a,i},Y_{a,i+1})-\gamma Y_{a,i},
\end{equation}
where $\gamma$ is the Dinkelbach variable~\cite{dinkelbach1967nonlinear}.
With this definition, the Dinkelbach reformulation is
\begin{align}\label{eq:dinkel}
    J(\gamma)=\min_{a\in[K]}\;\mathbb{E}[D(Y_{a,i},Y_{a,i+1},\gamma)].
\end{align}

Given the optimal node $a^*$ to Problem \eqref{eq:stationary_problem}, we have $\mathbb{E}[D(Y_{a,i},Y_{a,i+1},\gamma^*)]\geq0$ and $\mathbb{E}[D(Y_{a^*,i},Y_{a^*,i+1},\gamma^*)]=0$, indicating that $a^*$ is also optimal to Problem~\eqref{eq:dinkel}. This establishes the equivalence of Dinkelbach's transformation in the stationary policy space.
The following lemma~\cite[Theorem~1]{dinkelbach1967nonlinear}
guarantees the equivalence of the transformed problem.
\begin{lemma}\cite[Theorem 1]{dinkelbach1967nonlinear}\label{Lemma_Dinkel}
When $\gamma^*$ is the optimum objective value of Problem \eqref{eq:stationary_problem}, Problem \eqref{eq:dinkel} with $\gamma=\gamma^*$ is equivalent to Problem \eqref{eq:stationary_problem} and $J(\gamma^*)=0$. 
\end{lemma}

Lemma \ref{Lemma_Dinkel} shows the possibility of equivalently reformulating the fractional objective into one with an average cost. In addition, the Dinkelbach's method suggests that the optimum objective $\gamma^*$ of Problem \eqref{eq:stationary_problem} satisfies  $J(\gamma^*)=0$ and hence conducting a line search over $\gamma$ leads to the optimal arm $a^*$.

We note that the average-cost objective is essential for the best arm online learning that will be discussed in Section~\ref{sec:algorithm}.
Unlike traditional MAB formulations that rely on immediate rewards, our objective is tightly coupled with the underlying Markovian state transitions.
This dependency requires reasoning over long-term temporal trajectories under unknown node dynamics rather than single-step outcomes, which is significantly more challenging as it deviates from standard bandit assumptions and makes sampling strategies based on instantaneous rewards inapplicable.
\color{black}

\section{Best Arm Identification}
To identify the best node $a\in[K]$ that minimizes the long-term average AoI ${\Delta}^{(\text{ave})}$ with high confidence and sample efficiency, we formulate our problem as a BAI task under the fixed-confidence setting. 

\subsection{Definition of Age-Optimal Best Arm}
Following with the restless edge node model in Section \ref{sec:restless_server}, given the stationary distribution $\mu_{\theta_a}$ of arm $a$, Problem \eqref{eq:dinkel} is equivalent to
\begin{align}\label{eq:problem_stationary}
    \min_{a\in[K]} \;&\sum\limits_{s\in\mathcal{S}}\mu_{\theta_a}(s)\Big[\frac{1}{2}f^2(s)-\gamma f(s)\Big]\nn
    &+\sum\limits_{s,s'\in\mathcal{S}}\mu_{\theta_a}(s)P_{\theta_a}(s'|s)f(s)f(s').
\end{align}
To facilitate analysis, we decompose the objective in \eqref{eq:problem_stationary} into a transition-independent term $\eta_{\theta_a}(\gamma)$ and transition-dependent term $\Gamma_{\theta_a}$, given as
\begin{align}
    \eta_{\theta_a}(\gamma)&\triangleq\sum_{s\in\mathcal{S}}\left(\frac{1}{2}f^2(s)-\gamma f(s)\right)\mu_{\theta_a}(s), \quad a\in[K], \label{eq:mean}\\    \Gamma_{\theta_a}&\triangleq\sum_{s\in\mathcal{S}}f(s)f(s')P_{\theta_a}(s'|s)\mu_{\theta_a}(s),\quad a\in[K]. \label{eq:covariance}
\end{align}
Therefore, Problem \eqref{eq:problem_stationary} reduces to
\begin{align}\label{eq:aoi-bandit}
    \min_{a\in[K]}~g_{\theta_a}(\gamma)\triangleq\eta_{\theta_a}(\gamma)+\Gamma_{\theta_a}.
\end{align}
Given the Dinkelbach variable $\gamma$, we define the age-optimal best arm $a^*(\bs{\theta},\gamma)$ under instance $\bs{\theta}\triangleq[\theta_1,\ldots,\theta_K]^{\top}$ as
\begin{align}\label{eq:best_arm}
    a^*(\bs{\theta},\gamma)\triangleq\argmin\limits_{a\in[K]}\;  g_{\theta_a}(\gamma). 
\end{align}

Our goal is to identify $a^*(\bs{\theta}, \gamma)$ in a sample-efficient manner without any prior knowledge of the instance $\bs{\theta}$.
We formulate this BAI problem under a fixed-confidence setting, where the scheduler aims to identify the best arm with as few samples as possible while ensuring that the misidentification probability is limited to a pre-specified confidence level.

\subsection{Best Arm Identification Policy}
To identify the best arm, the scheduler sequentially selects edge nodes to process update packets and observes the resulting service times.
Let $A_n\in[K]$ denote the arm selected for the $n$th update and let $Y_{n}$ be the corresponding observed service time.
Note that we adopt the restless bandit setting, in which each arm evolves according to its internal Markovian dynamics, regardless of whether it is selected.
Let $\mathcal{F}_n$ 
be the $\sigma$-field generated by the sequence of past actions and observations up to update $n$, i.e.,
\begin{align}
    \mathcal{F}_n\triangleq\sigma\left(A_1,Y_{1},\ldots,A_n,Y_{n}\right).
\end{align}
A BAI policy under the fixed-confidence setting consists of:
\begin{itemize}
    \item \textit{Sampling rule} $\pi=\{\pi_n\}_{n\ge1}$: an $\mathcal{F}_{n-1}$-measurable rule that determines the arm $A_n$ to be selected for update $n$;
    \item \textit{Stopping rule} $\tau$: a stopping index adapts $\{\mathcal{F}_n\colon n\ge1\}$ indicating when to terminate the arm selection process;
    \item \textit{Decision rule} $a_{\tau}$: an $\mathcal{F}_{\tau}$-measurable rule that specifies the candidate best arm $a\in[K]$ at termination.
\end{itemize}
Given a confidence level $\delta\in(0,1)$, we define
\begin{align}\label{eq:set_policy}
    \Pi(\delta)\triangleq\Bigg\{(\pi,\tau,a_{\tau}):
    \begin{array}{l}
         \mathbb{P}_{\bs{\theta}}(a_{\tau}\neq a^*(\bs{\theta},\gamma))\leq\delta  \\
         \mathbb{P}_{\bs{\theta}}(\tau<\infty)=1 
    \end{array}
    \Bigg\}
\end{align}
as the set of all policies that terminate in finite time with probability one and upon termination outputs the best arm with probability at least $1-\delta$.
Our objective is to design a policy $\pi\in\Pi(\delta)$ that minimizes the expected sample complexity for solving Problem \eqref{eq:best_arm}, which is equivalent to characterize the value of
\begin{align}\label{eq:bai_sample}
    \inf_{\pi\in\Pi(\delta)}\, \mathbb{E}_{\bs{\theta}}[\tau_{\pi}].
\end{align}
We consider the regime in which $\delta\to0$.
However, the structural complexity of our AoI minimization problem originating from its dependence on long-term Markovian dynamics 
poses significant challenges for efficient algorithm design, as we discuss in the following section.


\section{Age-Optimal BAI Algorithm}\label{sec:algorithm}
We now describe the key ideas of our proposed age-optimal BAI algorithm which is a policy satisfying \eqref{eq:set_policy}.
The global structure of the age-optimal BAI algorithm
consists of three loops, including Markov Regeneration Sampling, Age-Aware LUCB procedure and Dinkelbach Update.
\begin{figure}[t]
    \centering
    \includegraphics[width=\linewidth]{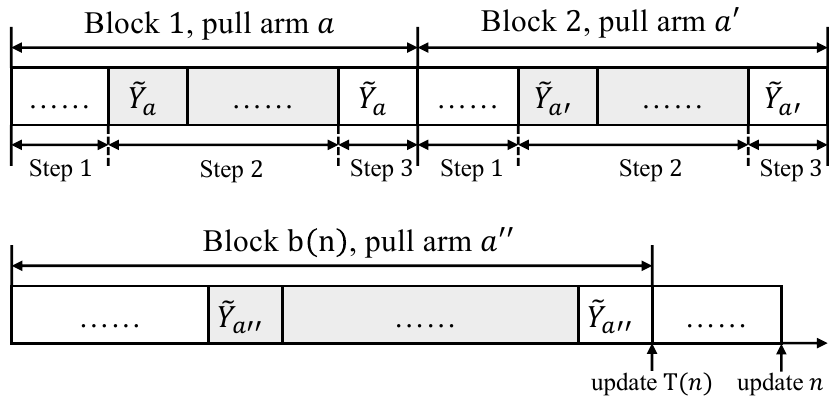}
    \caption{The structure of the Markov regeneration sampling strategy.}
    \label{fig:blocksample}
\end{figure}

\subsection{Markov Regeneration Sampling Strategy}
It follows from \eqref{eq:aoi-bandit} that our objective includes a transition-dependent term $\Gamma_{\theta_a}$, which requires tracking long-term temporal trajectories governed by unknown node congestion transitions.
To address this temporal dependence between consecutive samples, we propose a Markov regeneration sampling strategy inspired based on the regenerative process introduced in~\cite{smith1955regenerative}.

Since we consider the deterministic emission setting,
for each arm $a$, we define a regenerative state $\tilde{X}_a$ with associated observation $\tilde{Y}_a$, which marks both the beginning and end of a Markov regeneration process.


We define $b(n)$ as the number of completed blocks up to update $n$, and $T(n)$ as the update index marking the end of the most recently completed block.
As depicted in Fig.~\ref{fig:blocksample}, each sampling block $b\in\{1,2,\ldots\}$ consists of the following steps:
\begin{itemize}
    \item \textit{Initialization (Step 1)}: Pull arm $a$ repeatedly until the regenerative observation $\tilde{Y}_a$ is encountered;
    \item \textit{Regeneration Sampling (Step 2)}: Continue sampling arm $a$ and collecting observations until revisiting $\tilde{Y}_a$;
    \item \textit{Block Termination (Step 3)}: Close the block upon the second observation of $\tilde{Y}_a$.
\end{itemize}
We focus on the samples collected during the Markov regeneration process (\textit{Step $2$}) of each block.
Let $j_b$ denote the cumulative number of samples obtained in \textit{Step $2$} across the first $b$ blocks,
and $T_{j_b}^a$ represent the total number of such samples collected from arm $a$.
The corresponding observations are denoted by $Y_{a,1},Y_{a,2},\ldots,Y_{a,T_{j_b}^a}$, from which we construct the empirical estimators:
\begin{subequations}
\begin{align}
    \hat{\eta}_{a,T^a_{j_b}}(\gamma)&=\frac{1}{T^a_{j_b}}\sum\limits_{i=1}^{T^a_{j_b}}\left(\frac{1}{2}Y^2_{a,i}-\gamma Y_{a,i}\right),\label{eq:estimator_eta}\\
    \hat{\Gamma}_{a,T^a_{j_b}}&=\frac{1}{T^a_{j_b}}\sum\limits_{i=1}^{T^a_{j_b}}\left(Y_{a,i} Y_{a,i+1}\right).\label{eq:estimator_Gamma}
\end{align}
\end{subequations}
Hence, the empirical cost $\hat{g}_{\theta_a}(\gamma)$ of arm $a$ is given by
\begin{align}\label{eq:empirical age}
    \hat{g}_{\theta_a}(\gamma)=\hat{\eta}_{a,T^a_{j_b}}(\gamma)+\hat{\Gamma}_{a,T^a_{j_b}}.
\end{align}

To bootstrap the regeneration process, the first block for each arm skips \textit{Step 1} by treating the initial observation as the regenerative point. The overall procedure for executing Markov regeneration sampling and updating the empirical estimators is summarized in Algorithm~\ref{alg:rca}.
\begin{algorithm}[t]\label{alg:rca}
\caption{Markov Regeneration Sampling}
\KwIn{Dinkelbach’s variable $\gamma$, block index $b$, Step $2$ update index $j$\;
}
Given counters and estimators for arm $a$: $T^a_j$, $\hat{\eta}_a(\gamma)$ and $\hat{\Gamma}_a$ based on \eqref{eq:estimator_eta}, \eqref{eq:estimator_Gamma}\;
\Comment{Step 1: Initialization}
Play arm $a$ and observe $Y_{a}$\label{alg:rca_1}\;
\While{$Y_{a}\neq \tilde{Y}_{a}$}{
    $i\gets i+1$\;
    Play arm $a$ and observe $Y_{a}$\label{alg:rca_11}\;
}
\Comment{Step 2: Regeneration sampling}
$\hat{\eta}_{a}(\gamma)\gets \frac{\hat{\eta}_{a}(\gamma) T^a_j+(\frac{1}{2}\tilde{Y}^2_{a}-\gamma \tilde{Y}_{a})}{T^a_j+1}$, $\hat{\Gamma}_a\gets \frac{\hat{\Gamma}_a T^a_j+0}{T^a_j}$ \label{alg:rca_2}\;
$i\gets i+1$, $j\gets j+1$, $T^{a}_j\gets T^{a}_j+1$\;
Set $Y_p=\tilde{Y}_{a}$\;
Play arm $a$ and observe $Y_{a}$\;
\While{$Y_{a} \neq \tilde{Y}_{a}$}{
$\hat{\eta}_{a}(\gamma)\gets \frac{\hat{\eta}_{a}(\gamma) T^a_j+(\frac{1}{2}Y_{a}^2-\gamma Y_{a})}{T^a_j+1}$, $\hat{\Gamma}_{a}\gets \frac{\hat{\Gamma}_{a} (T^a_j-1)+Y_p Y_{a}}{T^a_j}$\;
$i\gets i+1$, $j\gets j+1$, $T^{a}_j\gets T^{a}_j+1$\;
Set $Y_{p}\gets Y_{a}$\;
Play arm ${a}$ and observe $Y_{a}$\label{alg:rca_22}\;
}
\Comment{Step 3: Block termination}
$\hat{\Gamma}_{a}\gets \frac{\hat{\Gamma}_{a} (T^a_j-1)+Y_p \tilde{Y}_{a}}{T^a_j}$\;
$b\gets b+1$, $i\gets i+1$ \label{alg:rca_end}\;
\end{algorithm}

\subsection{Age-Aware LUCB algorithm}
Based on the Markov regeneration sampling strategy, we now present the age-aware LUCB algorithm for identifying the age-optimal best arm with high confidence and sample efficiency.
A critical component in LUCB-type algorithms is the design of tight concentration bounds on the estimation, which support reliable confidence intervals that guide arm selection, comparison and stopping rules.

Unlike classical bandit settings with i.i.d. observations, our restless Markovian model involves arms whose underlying states evolve according to unknown and potentially non-reversible Markov chains.
This presents a fundamental challenge since standard concentration inequalities~\cite{chernoff1952measure,hoeffding1963probability} are no longer applicable due to the temporal correlations in the observations.

To obtain tight concentration bounds under Markovian dependence, it is essential to account for the mixing behavior of the chain, which captures the convergence rate to the stationary distribution.
To this end, we adopt the notion of the pseudo spectral gap~\cite{paulin2015concentration}, which extends spectral analysis to non-reversible Markov chains.
\begin{definition}[Pseudo Spectral Gap]
    Given the transition probability matrix $P_{\theta_a}$ of arm $a$ with stationary distribution $\mu_{\theta_a}$, let $P'_{\theta_a}$ denote its adjoint operator with respect to the inner product on the Hilbert space $L_2(\mu_{\theta_a})$.
    Then, the pseudo spectral gap of arm $a$ is defined as 
    \begin{align}\label{eq:pseudo}
        \beta_{\theta_a}=\max_{k\ge1}\,\{\beta'((P'_{\theta_a})^kP^k_{\theta_a})/k\}, 
    \end{align}
    where $\beta'(\cdot)$ denotes the spectral gap as defined in~\cite{lawler1988bounds}.
\end{definition}

To establish uniform concentration guarantees across all $K$ arms, we introduce the global pseudo spectral gap $\beta^{\min}\triangleq\min_{a\in[K]} \beta_{\theta_a}$, which characterizes the worst-case mixing rate among all Markovian arms under the restless bandit setting.
For quantifying initialization bias, we let $q_{\theta_a}$ denote the initial distribution of arm $a$ and define the vector $\bs{\omega}_a\triangleq(\frac{q_{\theta_a}(s)}{\mu_{\theta_a}(s)}, s\in\mathcal{S})$ to measure its deviation from stationarity.
Using the Minkowski inequality and using a global lower bound on the stationary probabilities $\mu^{\min}=\min_{a\in[K]}\min_{s\in\mathcal{S}}~\mu_{\theta_a}(s)$, we bound the deviation as
\begin{align}\label{eq:N_q}
    \Big\Vert\bs{\omega}_a\Big\Vert_2
    \le \sum\limits_{s\in\mathcal{S}}\Big\vert\frac{q_{\theta_a}(s)}{\mu_{\theta_a}(s)}\Big\vert\le\frac{1}{\mu^{\min}}.
\end{align}

Assuming $\mu^{\min}\ge \bar{\mu}$ for some constant $\bar{\mu}>0$, we provide a uniform bound on initialization error, which enables us to state the following concentration bound motivated by \cite[Theorem~3.3]{lezaud1998chernoff}.
\begin{proposition}\label{proposition:concentration}
    Given a constant $c\ge96(\bar{f}^2-\underline{f}^2)^2/\beta^{\min}$ and $K$ arms, for any $0<\delta\leq1$, it holds that
    \begin{align}\label{eq:concentration_bound}
        \mathbb{P}\left[\left|\hat{g}_{\theta_a}(\gamma)\!-\!g_{\theta_a}(\gamma)\right|\!<\!\sqrt{\frac{c\log{(4j_bK/(\delta\bar{\mu}))}}{T^a_{j_b}}}\right]\!\ge\!1\!-\!\frac{\delta}{K}.
    \end{align}
\end{proposition}
\begin{IEEEproof}
    See details in Appendix \ref{proof:concentration}.
\end{IEEEproof}

\begin{algorithm}[t]\label{alg:age-lucb}
\caption{Age-Aware LUCB}
\KwIn{Dinkelbach’s variable $\gamma$, confidence level $\delta$\;}
Initialize $b=1$, $i=0$, $i_2(b)=0$, $T^a_2(i_2(b))=0$, $\hat{\eta}_a=0$, $\hat{\Gamma}_a=0$ for all $a\in[K]$\;
\For{$b\le K$}{
    Play arm $b$ and obtain $\hat{\eta}_b^b(\gamma)$, $\hat{\Gamma}_b^b$ according to lines \ref{alg:rca_2}-\ref{alg:rca_end} of Algorithm~\ref{alg:rca}\;
}
\While{true}{
\For{$a=1,\ldots,K$}{
    Compute $\mathrm{LCB}_a(b)$ according to \eqref{eq:lcb_confidence}\label{alg:start_lucb}\;
}
Set $h(b)=\argmin_{a\in[K]}~\hat{\eta}_a^b(\gamma)+\hat{\Gamma}_a^b$ \label{alg:high_lucb}\;
Set $l(b)=\argmin_{a\neq h(b)} \mathrm{LCB}_a(b)$\label{alg:low_lucb}\;
Compute $\mathrm{UCB}_{h(b)}$ according to \eqref{eq:ucb_confidence}\label{alg:end_lucb}\;
\If{$\mathrm{UCB}_{h(b)}>\mathrm{LCB}_{l(b)}$\label{alg:condition}}{
\For{$a'\in\{h(b),l(b)\}$}{
    Play arm $a'$ and obtain $\hat{\eta}_{a'}^b(\gamma)$, $\hat{\Gamma}_{a'}^b$ according to Algorithm~\ref{alg:rca}\;
}
\Else{\Return{$a^*(\gamma)=h(b)$}}
}
}
\end{algorithm}

Following the concentration bounds established in Proposition~\ref{proposition:concentration}, we develop the age-aware LUCB selection scheme that integrates the Markovian regeneration sampling to enable reliable comparisons among arms with temporally correlated observations.

Different from traditional LUCB methods that assume i.i.d. samples and update confidence intervals after each pull, our approach introduces a block-based update mechanism that ensures statistical validity under Markovian dynamics.
Specifically, at each block index $b$, the algorithm computes the empirical estimates $\hat{\eta}_a^b(\gamma)$ and $\hat{\Gamma}_a^b$ for each arm $a$
and construct the lower and upper confidence bounds as follows:
\begin{subequations}\label{eq:confidence_interval}
\begin{align}
    \mathrm{LCB}_a(b)&=\hat{\eta}_a^b(\gamma)+\hat{\Gamma}_a^b-\mathrm{rad}_a(b),\label{eq:lcb_confidence} \\
    \mathrm{UCB}_a(b)&=\hat{\eta}_a^b(\gamma)+\hat{\Gamma}_a^b+\mathrm{rad}_a(b),\label{eq:ucb_confidence}
\end{align}
\end{subequations}
where $\mathrm{rad}_a(b)$ is the confidence radius given as
\begin{align}\label{eq:rad}
    \mathrm{rad}_a(b)=\sqrt{\frac{c\log{(4j_bK/(\delta\bar{\mu}))}}{T^a_{j_b}}}.
\end{align}

In each round $b$, as detailed in Lines \ref{alg:high_lucb} to \ref{alg:low_lucb} of Algorithm~\ref{alg:age-lucb}, the age-aware LUCB algorithm identifies two candidate arms: the current best arm $h(b)$, which minimizes the empirical cost estimate; and the challenger $l(b)$, which has the smallest lower confidence bound among the remaining arms.
These two arms are selected for further exploration to reduce uncertainty and tighten the confidence intervals.
The algorithm terminates when the upper confidence bound of the best arm falls below the lower confidence bound of its challenger, guaranteeing identification of the best arm with high confidence.

\subsection{Dinkelbach Update}
Leveraging Lemma~\ref{Lemma_Dinkel} and the monotonicity property that $J(\gamma)$ decreases in $\gamma$, we adopt the bisection method to search for the optimal Dinkelbach variable $\gamma^*$ such that $J(\gamma^*)=0$.
Following Problem \eqref{eq:stationary_problem}, when the algorithm
converges, the minimum average age achieved by the best arm is $\gamma^*$.

\section{Upper Bound on the Sample Complexity}
In this section, we present an upper bound on the sample complexity of the proposed age-optimal BAI strategy and further explore how key arm-dependent parameters affect the algorithm efficiency.

While Proposition~\ref{proposition:concentration} provides concentration guarantees for individual arm estimates, deriving a valid stopping rule requires confidence intervals to simultaneously hold across all arms and sampling blocks.
To this end, we establish the global high-probability \emph{good} event in the following lemma.
\begin{lemma}\label{lemma:event}
    Define the good event $\mathcal{E}$ as
    \begin{align}
        \mathcal{E}\triangleq\bigcap_{a\in[K]} \bigcap_{b\in\mathbb{N}}\left\{\left\vert \hat{\eta}_{a}^b(\gamma)+\hat{\Gamma}_{a}^b-g_{\theta_a}(\gamma) \right\vert \le \mathrm{rad}_a(b)\right\},
    \end{align}
it holds with high probability at least $1-\delta$, i.e., $\mathbb{P}\{\mathcal{E}\}\ge 1-\delta$.
\end{lemma}
\begin{IEEEproof}
    According to \eqref{eq:confidence_interval}, we prove that $\mathbb{P}\{g_{\theta_a}(\gamma)<\mathrm{UCB}_a(b)<g_{\theta_a}(\gamma)+2\mathrm{rad}_a(b)\}\ge1-\frac{\delta}{K}$ and $\mathbb{P}\{\mathrm{LCB}_a(b)<g_{\theta_a}(\gamma)<\mathrm{LCB}_a(b)+2\mathrm{rad}_a(b)\}\ge1-\frac{\delta}{K}$ holds for all blocks. Applying a union bound on the failure probabilities over all arms completes the proof.
\end{IEEEproof}
Lemma~\ref{lemma:event} ensures that, with probability at least $1-\delta$, the confidence intervals contain the true means throughout the learning process, thereby ensuring the correctness of the stopping rule.

Without loss of generality, we assume that arm $1$ is the unique best arm achieving the minimum average AoI.
Additionally, we assume $g_{\theta_1}(\gamma^*)<g_{\theta_2}(\gamma^*)\le g_{\theta_3}(\gamma^*)\le\ldots\le g_{\theta_K}(\gamma^*)$.
We define the optimality gap of each arm $a$ as $\Lambda_{\theta_a}=g_{\theta_a}(\gamma^*)-g_{\theta_1}(\gamma^*)$ where $\Lambda_{\theta_a}>0$.

The key to bounding the total number of measurements is to determine when a suboptimal arm can be confidently eliminated from further exploration.
We introduce a midpoint threshold $\xi=\frac{1}{2}(g_{\theta_1}(\gamma^*)+g_{\theta_2}(\gamma^*))$  and declare an arm as \emph{active} in the age-optimal BAI algorithm if
\begin{align}
    \mathrm{UCB}_{h(b)}(b)>\xi\quad\mbox{or}\quad\mathrm{LCB}_{l(b)}(b)<\xi,
\end{align}
indicating that the empirical best arm $h(b)$ has not yet been confirmed as optimal or the challenger $l(b)$ cannot yet be ruled out. Note that the midpoint threshold $\xi$ is introduced purely as an analytical tool in the sample complexity analysis. It is not required or computed during the execution of our proposed Age-Aware LUCB algorithm so it is parameter free.
The Age-Aware LUCB  algorithm terminates only when no arm remains \emph{active}, i.e., the best arm is identified with high confidence and all suboptimal arms have been excluded.

Under the high-probability \emph{good} event $\mathcal{E}$ from Lemma \ref{lemma:event}, the following corollary characterizes an instance-dependent condition for eliminating suboptimal arms.
\begin{corollary}\label{coro:1/4}
    Any suboptimal arm $a$ is no longer active once the confidence radius satisfies $\mathrm{rad}_a(b)\le\Lambda_{\theta_a}/4$.
\end{corollary}
\begin{IEEEproof}
    See details in Appendix \ref{proof:coro_1/4}.
\end{IEEEproof}
Corollary \ref{coro:1/4} indicates that each suboptimal arm is eliminated once its confidence interval becomes sufficiently tight.  
Note that under the Markovian regeneration sampling strategy, the total sample count for each arm $a$ in block $b$ includes initialization, regeneration and one additional block termination sample.
To quantify the initialization overhead, we adopt the notion of
the expected hitting time $\Psi_{\theta_a}^{s,s'}$ \cite{feller1991introduction} that denotes the expected steps to reach state $s'$ from state $s$ under the transition probability matrix $P_{\theta_a}$ of arm $a$.
For the time-homogeneous Markov chain considered in this paper, $\Psi_{\theta_a}^{s,s'}$ satisfies
\begin{align}
    \Psi_{\theta_a}^{s,s'}=1+\sum_{s''\in\mathcal{S}}P_{\theta_a}(s''|s)\Psi_{\theta_a}^{s'',s'}.
\end{align}

Since each Markovian regeneration block returns to a designated state drawn from the stationary distribution $\mu_{\theta_a}$, the expected regeneration length is lower bounded by the inverse of the minimal stationary probability $\mu^{\min}_{\theta_a}=\min_{s\in\mathcal{S}}\mu_{\theta_a}(s)$.
We define the worst-case mean hitting time as $\Psi_{\theta_a}^{\max}=\max_{s,s'}\Psi_{\theta_a}^{s,s'}$.
Leveraging the elimination condition in Corollary \ref{coro:1/4},
we provide an instance-dependent upper bound on the expected sample complexity in terms of the required number of blocks and per-block sampling cost.
\begin{theorem}[Sample Complexity]\label{theorem:upper_bound}
    Given 
    the confidence parameter $\delta$ and with probability at least $1-\delta$, 
    the expected stopping time $\mathbb{E}[\tau_\delta]$ satisfies
    \begin{align}\label{eq:upper}
        \mathbb{E}[\tau_\delta]=\mathcal{O}\bigg(\sum_{a: \Lambda_{\theta_a}>0}\frac{\frac{1}{\mu^{\min}_{\theta_a}}+\Psi^{\max}_{\theta_a}}{\Lambda^2_{\theta_a}}\log{\Big(\frac{1}{\delta\bar{\mu}}\Big)}\bigg)\quad\mbox{as}\quad\delta\to0.
    \end{align}
\end{theorem}
We observe from Theorem~\ref{theorem:upper_bound} that the sample complexity scales logarithmically with $1/\delta$, indicating that higher confidence levels 
requires more samples. 
In addition, the $\mathcal{O}(\Lambda_{\theta_a}^{-2})$ dependence captures the difficulty of identifying the best arm, where smaller optimality gaps require greater exploration. 
Furthermore, the Markovian dynamics introduce an additional sampling overhead $\Psi^{\max}_{\theta_a}$ associated with the Markov mixing behavior, reflecting that slower mixing delays the confidence updates and increases sample cost.



\section{Lower Bound}
In this section, we investigate the fundamental limits of the proposed age-optimal BAI problem by deriving an information-theoretic lower bound on the sample complexity. 
Since the general expression of the lower bound is implicit, we further present a case study that yields an explicit form, which enables a direct comparison with the upper bound stated in Theorem \ref{theorem:upper_bound}.

\subsection{Information-Theoretic Lower Bound}
To characterize the inherent difficulty of our age-optimal BAI problem,
we derive an information-theoretic lower bound on the sample complexity that holds for any algorithm.
We define $\mathrm{ALT}(\bs{\theta})$ as the set of all alternative problem instances that represent distinct system configurations, in which the best edge node differs from that under $\bs{\theta}$.
Each $\bs{\theta}'\in\mathrm{ALT}(\bs{\theta})$ satisfies
$a^*(\bs{\theta}')\neq a^*(\bs{\theta})$.
In other words, $\bs{\theta}'\in\mathrm{ALT}(\bs{\theta})$ satisfies $g_{\theta'_{a^*(\bs{\theta}')}}<g_{\theta_{a^*(\bs{\theta})}}$.

The log-likelihood ratio of observations up to update $n$ under instances $\bs{\theta}$ versus $\bs{\theta}'$ and denoted by $L_{\bs{\theta},\bs{\theta}'}(n)$ is then given by
\begin{align}\label{eq:log-likeli}
    L_{\bs{\theta},\bs{\theta}'}(n)=\log\frac{P_{\bs{\theta}}(A_{1:n},Y_{1:n})}{P_{\bs{\theta}'}(A_{1:n},Y_{1:n})}.
\end{align}

Due to the restless dynamics, we introduce the notion of delay in pulling arm $d_a\ge1$ as in \cite{Karthik2023best}, which is defined as the update elapsed since the arm $a$ was last selected.
Let $d_a(n)$ and $s_a(n)$ denote the delay and last observed state of arm $a$ for update $n$. 
We denote the state space of this Markov Decision Process as $\mathbb{S}$.
The system global state is then given as $(\bs{s}(n),\bs{d}(n))\in\mathbb{S}$, where $\bs{s}(n)\triangleq(s_1(n),\ldots,s_K(n))$ and $\bs{d}(n)\triangleq(d_1(n),\ldots,d_K(n))$.
The transition probabilities under instance $\bs{\theta}$ follow
\begin{align}\label{eq:TPM_d}
    &\mathbb{P}_{\bs{\theta}}\left((\bs{s}(n+1),\bs{d}(n+1))=( \bs{s}',\bs{d}')\mid(\bs{s}(n),\bs{d}(n),A_n=a\right)\nn
    &\quad=
    \begin{cases}
    P_{\theta_a}^{d_a}(s_a'\mid s_a), & \text{if}\,d_a'=1,d_{a'}'=d_{a'}+1,\ \\
    &s_{a'}'=s_{a'}\ \forall a' \neq a, \\
    0, & \text{otherwise}.
    \end{cases}
\end{align}

Notably, $\mathbb{S}$ is an infinite state space. For further analysis, we reduce $\mathbb{S}$ to a finite state space by constraining the delay of each arm to a positive integer $D$, yielding the truncated state space $\mathbb{S}_D\subset\mathbb{S}$. 
Let $\mathbb{S}_{D,a}\subset\mathbb{S}_D$ denote the subset where $d_a=D$,
for $(\bs{s},\bs{d})\in\mathbb{S}_{D,a}$, the transition probabilities satisfying 
\begin{align}
    &\mathbb{P}_{\bs{\theta}}\left((\bs{s}(n+1),\bs{d}(n+1))=(\bs{s}',\bs{d}')\mid(\bs{s}(n),\bs{d}(n)),A_n=a\right)\nn
    &\quad=
    \begin{cases}
    P_{\theta_a}^{D}(s_a'\mid s_a), & \text{if}\,d_a'=1,d_{a'}'=d_{a'}+1,\ \\
    &s_{a'}'=s_{a'}\ \forall a' \neq a, \\
    0, & \text{otherwise}.
    \end{cases}
\end{align}
Combining the dynamics in \eqref{eq:TPM_d} for all $(\bs{s},\bs{d})\in\mathbb{S}_D \setminus \mathbb{S}_{D,a}$,
we define the transition kernel as
\begin{align}
    Q_{\bs{\theta},D}(\bs{s}',\bs{d}'|\bs{s},\bs{d},a)&\triangleq\mathbb{P}_{\bs{\theta}}((\bs{s}(n+1),\bs{d}(n+1))=(\bs{s}',\bs{d}')\nn
    &\qquad\mid(\bs{s}(n)),\bs{d}(n),A_n=a).
\end{align}

We assume each arm is selected once for the first $K$ updates to obtain the initial observations.
For subsequent updates $n>K$, we denote the state visitations of $(\bs{s},\bs{d})$ up to update $n$ of arm $a$ as $N((\bs{s},\bs{d}),a,n)$ and the
number of transitions from $(\bs{s},\bs{d})\in\mathbb{S}_D$ to $\bs{s}'$ by $N((\bs{s},\bs{d}),\bs{s}',a,n)$.
Following \eqref{eq:log-likeli}, the log-likelihood can be represented as
\begin{align}\label{eq:log-likeli-2}
    &L_{\bs{\theta},\bs{\theta}'}(n)=\sum_{a=1}^K \log \frac{q_{\theta_a}(Y_{a,1})}{q_{\theta'_a}(Y_{a,1})}\nn
    &+\sum_{a=1}^K \sum_{\bs{s}'\in\mathcal{S}}\sum\limits_{(\bs{s},\bs{d})\in\mathbb{S}_D}\!\!\!N((\bs{s},\bs{d}),\bs{s}',a,n)\log\frac{P^{d_a}_{\theta_a}(s_a'|s_a)}{P^{d_a}_{\theta'_a}(s_a'|s_a)}.
\end{align}
Applying a change-of-measure argument, we establish the condition for any algorithm to identify the best arm with confidence at least $1-\delta$, given by
\begin{align}\label{eq:change-of-measure}
    \inf_{\bs{\theta}'\in\mathrm{ALT}(\bs{\theta})}~\mathbb{E}_{\bs{\theta}}[L_{\bs{\theta},\bs{\theta'}}(n)]\ge \mathrm{KL}(\delta\Vert1-\delta), \forall n>K,
\end{align}
where $\mathrm{KL}(\delta\Vert1-\delta)$ denotes the KL divergence between the Bernoulli distributions with parameters $\delta$ and $1-\delta$.

To facilitate the derivation of a tight lower bound from \eqref{eq:change-of-measure}, we introduce a normalized occupation measure $\nu(\bs{s},\bs{d},a)$, which describes the steady-state distribution of arm selections after the initial exploration:
\begin{align}
    \nu(\bs{s},\bs{d},a)=\frac{\mathbb{E}_{\bs{\theta}}[N((\bs{s},\bs{d}),a,n),\tau]}{\mathbb{E}_{\bs{\theta}}[\tau-K]}.
\end{align}
Let $\Sigma_D(\bs{\theta})$ be the space of all probability mass functions $\nu$ and $Q_{\bs{\theta},D}(\cdot\mid\bs{s},\bs{d},a)\triangleq [Q_{\bs{\theta},D}(\bs{s}',\bs{d}'\mid\bs{s},\bs{d},a)]^\top$,
we characterize the information-theoretic lower bound for \eqref{eq:bai_sample} as follows.
\begin{theorem}[Lower Bound]\label{theorem:lower_bound}
    For any $\bs{\theta}\in\Theta^K$ and $D\ge1$, the expected stopping time of any algorithm that is $\delta$-PAC satisfies
    \begin{align}
        \mathbb{E}[\tau_\delta]=\Omega\left(\frac{1}{T_D(\bs{\theta})}\log\left(\frac{1}{\delta}\right)\right)\quad\mbox{as}\quad\delta\to0,
    \end{align}
    where $T_D(\bs{\theta})$ is given by
    \begin{align}
        T_D(\bs{\theta})=&\sup_{\nu\in\Sigma_D(\bs{\theta})}\inf_{\bs{\theta}'\in\mathrm{ALT}(\bs{\theta})}\sum_{(\bs{s},\bs{d})\in\mathbb{S}_D}\sum_{a=1}^{K}\nu(\bs{s},\bs{d}, a)\nn
        &\;\;\times\mathrm{KL}\left(Q_{\bs{\theta},D}(\cdot\mid\bs{s},\bs{d},a)\Vert Q_{\bs{\theta}',D}(\cdot\mid\bs{s},\bs{d},a)\right).
    \end{align}
\end{theorem}
Theorem \ref{theorem:lower_bound} measures the fundamental difficulty of the age-optimal BAI problem across system configurations, where a smaller $T_D(\bs{\theta})$ indicates more similar egde node dynamics, thereby requiring more samples to achieve reliable identification.
\subsection{Case Study}
To provide further insight into the fundamental limits characterized by Theorem \ref{theorem:lower_bound}, we consider a concrete and analytically tractable two-arm restless bandit system and fix the delay constraint $D=2$.
Each arm evolves as a two-state Markov chain on the state space $\mathcal{S}=\{0,1\}$, with deterministic emission function $f(s)=s$. 
The transition dynamics of each arm is governed by the one-parameter exponential family defined in \eqref{eq:tilde_P}.
Specifically, we adopt a base transition matrix $P$ satisfying Assumption \ref{assump:P}, given by
\begin{align}
    P=\begin{pmatrix}
        1-p & p \\
        q& 1-q
    \end{pmatrix}, \quad p,q\in(0,1).
\end{align}
Following the normalization in \eqref{eq:P_theta}, the stationary probability of state $1$ for arm $a$ parameterized by $\theta_a$, is given by
\begin{align}\label{eq:mu(1)}
    \mu_{\theta_a}(1)=\frac{pe^{\theta_a}(q+(1-q)e^{\theta_a})}{pe^{\theta_a}(q+(1-q)e^{\theta_a})+q(1-p+pe^{\theta_a})}.
\end{align}

Assume that arm $1$ is the unique optimal arm under the true instance $\bs{\theta}=(\theta_1,\theta_2)$. 
Since a larger $\theta_a$ increases $\mu_{\theta_a}(1)$ according to \eqref{eq:mu(1)}, and hence leads to a larger average AoI, this implies $\theta_1<\theta_2$.
To derive an explicit lower bound, we construct an alternative instance $\bs{\theta}'=(\theta'_1, \theta'_2)$ such that $\theta'_1=\theta_1$ and $\theta'_2 = \theta_2-\varepsilon$, where $\varepsilon$ is chosen so that arm $1$ becomes suboptimal under $\boldsymbol{\theta}'$. 
Under this construction, following \eqref{eq:log-likeli-2}, the log-likelihood ratio simplifies to
\begin{align}
    \hspace{-0.16cm} 
    L_{\bs{\theta},\bs{\theta}'}(n)\!=\! \sum_{\bs{s}'\in\mathcal{S}}\sum\limits_{(\bs{s},\bs{d})\in\mathbb{S}_D}\!\!N((\bs{s},\bs{d}),\bs{s}',n)\log\frac{P^{d_2}_{\theta_2}(s'|s)}{P^{d_2}_{\theta_2-\varepsilon}(s'|s)},
\end{align}
since arm $1$ has identical dynamics under both instances.

Therefore, the difficulty of distinguishing $\boldsymbol{\theta}$ from $\boldsymbol{\theta}'$ is dominated by the KL divergence between the trajectory distributions induced by $\theta_2$ and $\theta_2-\varepsilon$.
For the exponential family, the KL divergence associated with a trajectory of length $d_2$ admits a second-order Taylor expansion in $\varepsilon$, given by
\begin{align}
    \mathrm{KL}(P^{d_2}_{\theta_2}(\cdot|s)\Vert P^{d_2}_{\theta_2-\varepsilon}(\cdot|s))=\frac{\varepsilon^2}{2}\mathcal{I}_{\theta_2}^{(d_2)}+o(\varepsilon^2),
\end{align}
where $\mathcal{I}_{\theta_1}^{(d_2)}$ denotes the Fisher Information contained in $d_2$ consecutive observations. 
Under the exponential family parameterization, this Fisher Information is equivalent to the variance of the sufficient statistic, given as
\begin{align}
    \mathcal{I}_{\theta_2}^{(d_2)}=\mathrm{Var}_{\mu_{\theta_2}}\Big(\sum_{i=1}^{d_2} f(X_i)\Big).
\end{align}

Recall that we fix the delay constraint $D=2$, the maximal delay $d_2=2$ dominates the information accumulation process. Consequently, it suffices to analyze the Fisher information corresponding to $d_2=2$.
In this case, the Fisher information can be decomposed as
\begin{align}\label{eq:fisher-2}
    \mathcal{I}_{\theta_2}^{(2)}\!=\!\mathrm{Var}_{\mu_{\theta_2}}(X_1)\!+\!\mathrm{Var}_{\mu_{\theta_2}}(X_1)\!+\!\mathrm{Cov}_{\mu_{\theta_2}}(X_1+X_2).
\end{align}
Since $f(X_i)$ is Bernoulli with parameter $\mu_{\theta_2}(1)$ under stationarity, we have $\mathrm{Var}_{\mu_{\theta_2}}(X_1)=\mathrm{Var}_{\mu_{\theta_2}}(X_2)\triangleq \mu_{\theta_2}(1)(1-\mu_{\theta_2}(1))$. 
Moreover, the covariance is governed by the second largest eigenvalue $\lambda_2^{(2)}$ of the transition matrix $P_{\theta_2}$, yielding $\mathrm{Cov}_{\mu_{\theta_2}}(X_1+X_2)= \mu_{\theta_2}(1)(1-\mu_{\theta_2}(1))\lambda_2^{(2)}$. 
Substituting into \eqref{eq:fisher-2} gives
\begin{align}
    \mathcal{I}_{\theta_2}^{(2)}=\mu_{\theta_2}(1)(1-\mu_{\theta_2}(1))(2+2\lambda_2^{(2)}).
\end{align}

Furthermore, under the delay constraint $D=2$, any admissible sampling strategy is forced to select each arm at least once every three steps to prevent the delay from exceeding the limit. This implies that the normalized occupation measure satisfies $\nu(a)\ge \frac{1}{3}$.
Combining with the change-of-measure inequality in \eqref{eq:change-of-measure}, we obtain
\begin{align}\label{eq:case_start}
    \mathbb{E[\tau_\delta]}\ge\frac{6}{\varepsilon^2\mu_{\theta_2}(1)(1-\mu_{\theta_2}(1))(2+2\lambda_2^{(2)})}\log\Big(\frac{1}{\delta}\Big).
\end{align}

The lower bound derived above is expressed in terms of the perturbation parameter $\varepsilon$.
To relate it to the intrinsic difficulty of the problem instance, we next characterize the quantitative relationship between the perturbation magnitude $\varepsilon$ and the optimality gap $\Lambda_{\theta_2}$ in the following lemma.
\begin{lemma}\label{lemma:varepsilon}
    The perturbation $\varepsilon$ scales linearly with the optimality gap $\Lambda_{\theta_2}$ such that
    \begin{align}
        \varepsilon=\Omega(\Lambda_{\theta_2}).
    \end{align}
\end{lemma}
\begin{IEEEproof}
    Since the objective function $g_{\theta_a}(\gamma)$ is continuously differentiable with respect to $\theta_a$, we perform a first-order Taylor expansion of $g_{\theta_2-\varepsilon}(\gamma^*)$ around $\theta_1$, given by
    \begin{align}
        g_{\theta_2-\varepsilon}(\gamma^*)=g_{\theta_2}(\gamma^*)-\varepsilon\frac{\partial g_{\theta_a}(\gamma^*)}{\partial\theta_a}\Big\vert_{\theta_a=\theta_2}+o(\varepsilon)
    \end{align}
    For arm $1$ to not be optimal under the alternative instance $\bs{\theta}'$, it is necessary that
    \begin{align}
        g_{\theta_2-\varepsilon}(\gamma^*)\le g_{\theta_1}(\gamma^*).
    \end{align}
    Rearranging this inequality and neglecting higher-order terms yields
    \begin{align}
        \varepsilon\ge\frac{g_{\theta_2}(\gamma^*)-g_{\theta_1}(\gamma^*)}{\Big\vert\frac{\partial g_{\theta_2}(\gamma^*)}{\partial\theta_2}\Big\vert}=\frac{\Lambda_{\theta_2}}{\Big\vert\frac{\partial g_{\theta_2}(\gamma^*)}{\partial\theta_2}\Big\vert}.
    \end{align}
    Since $\frac{\partial g_{\theta_2}(\gamma^*)}{\partial\theta_2}$ is finite and nonzero, this establishes that $\varepsilon = \Omega(\Lambda_{\theta_2})$. This completes the proof.
\end{IEEEproof}

Lemma~\ref{lemma:varepsilon} indicates that the minimal perturbation required to alter the identity of the optimal arm is proportional to the optimality gap. Substituting this linear relationship into the inequality \eqref{eq:case_start} yields an explicit instance-dependent lower bound, summarized in the following remark.
\begin{remark}\label{remark:lowerbound}
    Fix $D=2$. For a two-arm, two-state restless bandit system, the expected stopping time of any $\delta$-PAC algorithm satisfies
\begin{align}
    \mathbb{E[\tau_\delta]}=\Omega\Bigg(\frac{1}{\Lambda_{\theta_2}^2(1+\lambda_2^{(2)})}\log\left(\frac{1}{\delta}\right)\Bigg).
\end{align}
\end{remark}
Remark~\ref{remark:lowerbound} recovers the classical $\Lambda_{\theta_2}^{-2}$ gap dependence and further reveals the role of temporal correlation through the factor $1+\lambda_2^{(2)}$.
Although $\lambda_2^{(2)}\to\pm 1$ both imply a vanishing spectral gap and hence slow mixing, the resulting correlation structures lead to fundamentally different statistical effects.

When $\lambda_2^{(2)}\!\to\!-1$, oscillatory dynamics create strong negative correlations, driving $(1+\lambda_2^{(2)})\!\to\!0$ and increasing the lower bound, 
which is consistent with 
the large hitting-time overhead $\Psi^{\max}_{\theta_a}$ in the upper bound under slow mixing.

When $\lambda_2^{(2)}\!\to\!1$, persistent dynamics yield positive correlations that enhance short-horizon information, allowing the lower bound to decrease despite slow mixing.
This does not contradict the upper bound in which $\Psi^{\max}_{\theta_a}$ reflects the algorithmic cost of enforcing effective decorrelation via regeneration under slow mixing, whereas the lower bound captures the intrinsic statistical difficulty without imposing such decorrelation requirements.

\section{Numerical Results}
In this section, we present numerical evaluations to demonstrate the empirical performance of our proposed age-optimal BAI scheme compared with three benchmarks.

\subsection{Setup}\label{subsec:numerical_design}


We consider a mobile edge computing system with $K=5$ edge nodes. 
Each edge node is modeled as an arm in a restless multi-armed bandit framework, where the state represents the background congestion level induced by co-located services and users.
Specifically, the congestion state of each node evolves over the finite state space $\mathcal{S}=\{0,1,2,3,4\}$ according to an unknown Markov process.

Following the deterministic emission abstraction adopted in our system model, each congestion state corresponds to a quantized end-to-end service delay experienced by an update packet.
We model the service delay as a monotone function of the congestion state by introducing a baseline networking delay $d_{\mathrm{net}}$ and an incremental computation delay $\zeta$.
Here, $d_{\mathrm{net}}$ represents the combined transmission and processing delay under negligible background congestion,
while $\zeta$ captures the additional computation overhead
induced by each incremental congestion level.
Accordingly, the observed service delay at each state $s\in\mathcal{S}$ is given by
\begin{align}\label{eq:setup}
    f(s)=d_{\mathrm{net}}+s\cdot\zeta.
\end{align}
In our numerical evaluation, we set $d_{\mathrm{net}}=10\,\mathrm{ms}$, which is on the order of round-trip latency reported for modern 5G access networks under favorable conditions~\cite{6824752}, and set $\zeta=5\,\mathrm{ms}$ to yield a realistic range of service delays for edge computing applications~\cite{8016573}.

Therefore, larger congestion states correspond to longer service delays based on \eqref{eq:setup}.
Although service delays may exhibit additional randomness in practice systems, such a quantized delay abstraction preserves the relative impact of congestion on AoI and serves as a first-order approximation suitable for performance evaluation.

To evaluate the sample efficiency of identifying the edge node that minimizes the time-average Age of Information, we compare the proposed age-optimal BAI algorithm with three benchmark variants: Markovian-LUCB, Markovian-UCB-BAI, and Markovian-Uniform Sampling.
All benchmarks adopt the Markovian regeneration sampling technique to address the structured AoI objective under restless dynamics.
Specifically, Markovian-LUCB generalizes the classical LUCB algorithm~\cite{kaufmann2016complexity} by incorporating confidence bounds over regenerative samples.
Markovian-UCB-BAI extends UCB algorithm to BAI setting by combining LUCB-style elimination to ensure sufficient exploration.
Markovian-Uniform Sampling follows a non-adaptive round-robin strategy, selecting arm $b\bmod K$ at each block $b$ regardless of performance estimates.

\begin{figure*}[t]
    \centering
    \subfigure[]{
        \centering
        \includegraphics[width=0.31\linewidth]{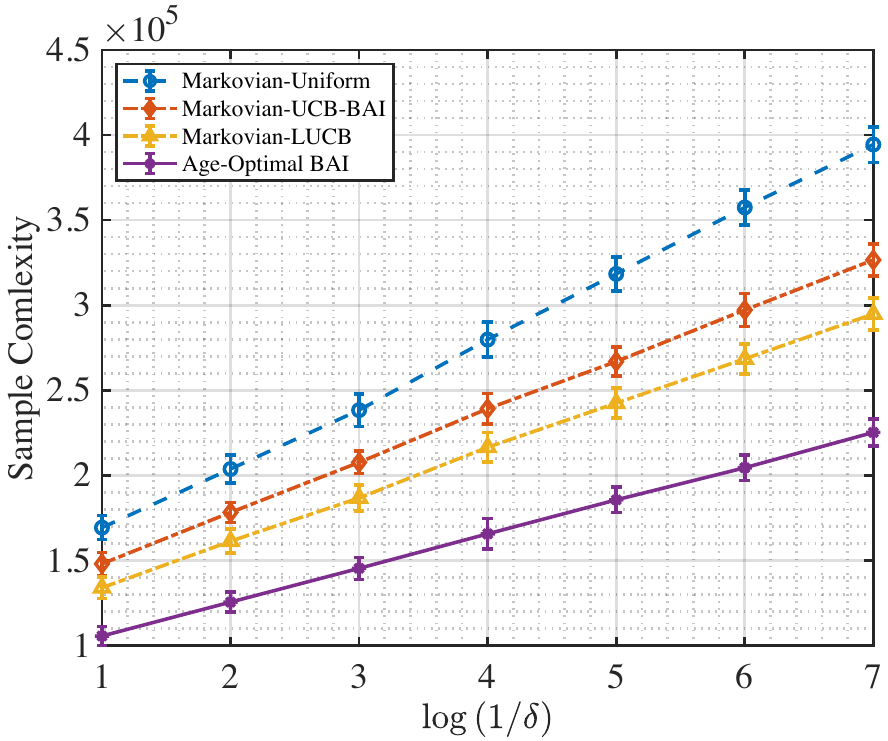}
        \label{fig:compare_upper_delta}
    }
    \subfigure[]{
        \centering
        \includegraphics[width=0.31\linewidth]{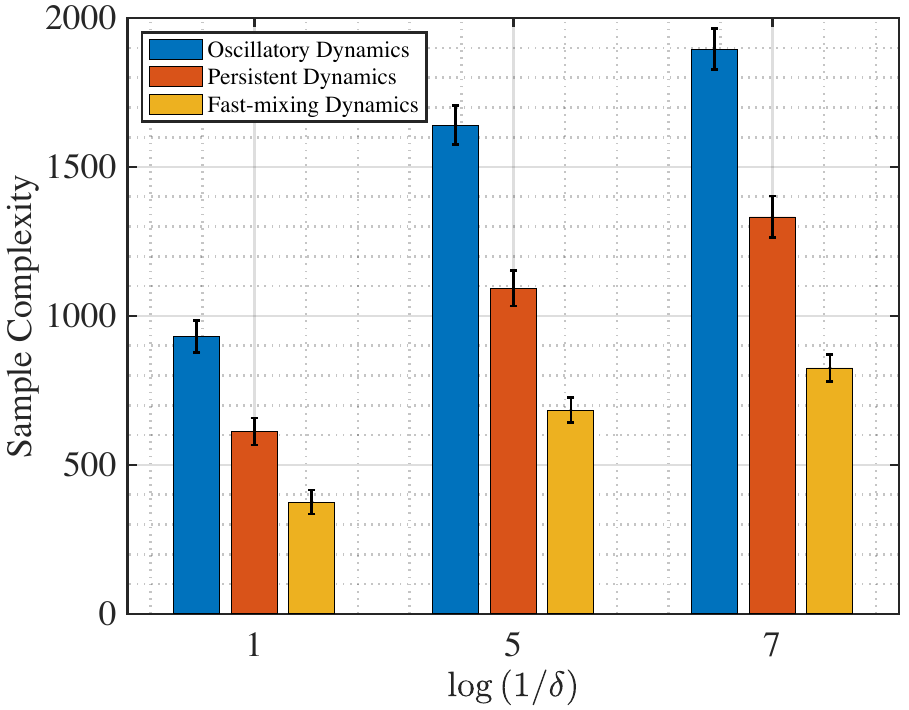}
        \label{fig:compare_mixing}
    }
    \subfigure[]{
        \centering
        \includegraphics[width=0.31\linewidth]{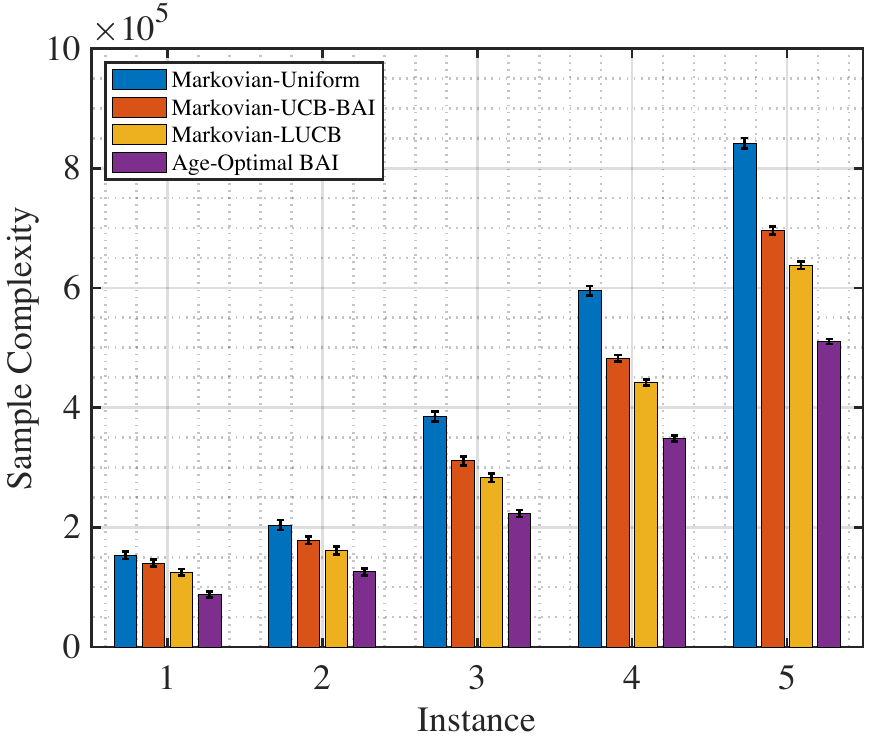}
        \label{fig:compare_upper_instance}
    }
\caption{Numerical evaluations of the performance comparison with (a) the confidence level $\delta$, (b) the Markov mixing behavior, and (c) the instances.
}
\label{fig:total}
\end{figure*}

\subsection{Performance Comparison via Confidence Levels}\label{subsec:compare_delta}

We interpret $\theta_a$ as a load-intensity parameter that biases the congestion-state Markov dynamics of edge node $a$ toward higher-delay states.
Specifically, a larger value of $\theta_a$ increases the likelihood of the node occupying more congested states, resulting in a larger long-term average AoI.
In this subsection, we start with a representative problem instance $\bs{\theta}=[0.1,0.3,0.5,0.7,0.9]$, which captures a heterogeneous edge computing
environment with progressively increasing background congestion levels across edge nodes.

To evaluate the performance gains of our proposed age-optimal BAI scheme, we compare the sample complexity of the other three benchmarks with the confidence level $\delta$ varying from $\delta=10^{-1}$ to $\delta=10^{-7}$ (i.e., $\log(1/\delta)$ from $1$ to $7$).
For each $\delta$, we run $1000$ trials to estimate the average sample complexity and construct the corresponding $95\%$ confidence intervals.
We plot the curves along with the error bars to show the sample complexity for all four methods in Fig~\ref{fig:compare_upper_delta}.
As expected, the sample complexity is increasing in $\log(1/\delta)$ for all four schemes, which is consistent with the upper bound behavior as stated in Theorem \ref{theorem:upper_bound}.
We next observe that the age-optimal BAI strategy significantly saves the sample complexity compared to the other methods.
For example, when $\log(1/\delta)=2$, age-optimal BAI achieves the sample complexity that is $38\%$ lower than Markovian-Uniform sampling, $29\%$ lower than Markovian-UCB-BAI and $22\%$ lower than Markovian-LUCB.
This saving expands with increasing $\log(1/\delta)$ to $43\%$, $31\%$ and $24\%$ respectively at $\log(1/\delta)=7$.
\textit{That is, compared to the other three benchmarks, our proposed age-optimal BAI strategy achieves greater sample efficiency when a stricter confidence level is required.}

\subsection{Performance Comparison via Markov Mixing}
To illustrate the impact of Markov mixing behavior on sample complexity, we consider a two-arm, two-state case study consistent with the lower-bound analysis.
Across all instances, arm $1$ is fixed, while arm $2$ is constructed to exhibit different mixing behaviors by varying its transition matrix.
Specifically, we consider three instances with $P_{\theta_2}=[0.01\,0.99;0.99\,0.01]$, $P_{\theta'_2}=[0.99\,0.01;0.01\,0.99]$ and $P_{\theta'_2}=[0.5\,0.5;0.5\,0.5]$.
These matrices correspond to oscillatory dynamics, persistent dynamics, and fast-mixing dynamics, respectively.

Fig.~\ref{fig:compare_mixing} plots the sample complexity of the proposed age-optimal BAI algorithm under varying confidence levels $\delta$.
We observe that the oscillatory instance incurs the highest sample complexity, consistent with the inflation predicted by the lower bound when $1+\lambda_2\to0$.
The persistent instance also requires more samples than the fast-mixing case, aligning with the upper bound through the increased overhead $\Psi^{\max}_{\theta_a}$ induced by slow mixing.

\subsection{Performance Comparison via Instances}
To characterize the robustness of our proposed scheme across varying system scenarios, we design five representative instances with increasing difficulty in identifying the age-optimal edge node.
In instance 1, we consider a stepped configuration with $\bs{\theta}(1)=[0.3,0.7,0.7,0.7,0.7]$, which models a heterogeneous edge environment where a single lightly loaded node coexists with several highly congested nodes.
For instances $\kappa=\{2,3,4,5\}$, we fix the optimal node at $\theta_1(\kappa)=0.1$ and gradually reduce the load-intensity gap between the
optimal and suboptimal nodes by defining $\theta_a(\kappa)=0.1+\frac{0.2(a-1)}{\kappa-1}$.
For example, instance $3$ corresponds to $\bs{\theta}(3)=[0.1,0.2,0.3,0.4,0.5]$.
As $\kappa$ increases, the congestion levels of the suboptimal nodes become
increasingly similar to that of the optimal node, leading to progressively smaller AoI gaps.
These instances capture realistic edge computing scenarios in which multiple
edge nodes exhibit comparable long-term performance due to similar background
traffic conditions.
Consequently, distinguishing the best node becomes increasingly challenging,
allowing us to systematically evaluate the sample efficiency and robustness of
the proposed age-optimal BAI strategy under varying degrees of difficulty.

Using the setup from Section \ref{subsec:numerical_design}, we fix the confidence level $\delta=0.01$ and compare the sample complexities of all four algorithms in Fig.~\ref{fig:compare_upper_instance}.
Our results show that the proposed age-optimal BAI strategy consistently achieves the lowest sample complexity across all instances.
\textit{Moreover, this advantage becomes increasingly significant as the best arm becomes harder to distinguish.}
For example, in instance 5, the age-optimal BAI scheme achieves 
$20\%$ reduction compared to Markovian-LUCB.


\section{Conclusion}
In this paper, we address the age-optimal BAI problem within a fixed-confidence setting, using the framework of RMABs. Our goal is to identify the best arm with high confidence minimum number of sample. To this end, we introduce an age-aware LUCB algorithm that incorporates a Markovian sampling technique. We also derive an upper bound on the sample complexity, which is influenced by the Markov chain's mixing behavior. Furthermore, we establish a fundamental information-theoretic lower bound characterized by a parameter $D$, which promotes diversity in arm selection.
A two-arm, two-state restless bandit case study of the lower bound reveals that the fundamental sample complexity scales quadratically with the inverse optimality gap and is fundamentally shaped by the temporal correlation of the underlying Markov dynamics.
This insight complements the analysis of the upper bound and highlights the intrinsic role of Markovian dependence in age-optimal bandit learning.

Our results demonstrate that the proposed scheme significantly reduces sampling costs, especially under stricter confidence requirements and in mobile edge computing systems with small performance gaps. As the first study to investigate the age-optimal BAI problem under unknown restless edge node dynamics, this work opens several promising avenues for future research. These include extending the framework to Hidden Markov Models with stochastic emissions and developing optimal waiting strategies.

\appendices
\section{Proof of Proposition~\ref{proposition:concentration}}\label{proof:concentration}
We prove Proposition~\ref{proposition:concentration} by establishing concentration bounds for the two components of the objective function $g_{\theta_a}(\gamma)$ separately, and then combining them via a union bound.

Recall that $g_{\theta_a}(\gamma)=\eta_{\theta_a}(\gamma)+\Gamma_{\theta_a}$, where $\eta_{\theta_a}(\gamma)$ and $\Gamma_{\theta_a}$ are estimated using regenerative samples. We begin by constructing two normalized centered functions over the regenerative samples, aligned with the empirical estimators defined in \eqref{eq:estimator_eta} and \eqref{eq:estimator_Gamma}, given by
\begin{align}
    &\tilde{f}_{\theta_a}(Y_{a,i})=\frac{\frac{1}{2}Y^2_{a,i}-\gamma Y_{a,i}-\eta_{\theta_a}(\gamma)}{(\frac{1}{2}\bar{f}^2-\gamma \bar{f})-(\frac{1}{2}\underline{f}^2-\gamma \underline{f})},\\
    &\hat{f}_{\theta_a}(Y_{a,i},Y_{a,i+1})=\frac{Y_{a,i} Y_{a,i+1}-\Gamma_{\theta_a}}{\bar{f}^2-\underline{f}^2}.
\end{align}
Due to the compact support of $Y_{a,i}\in[\underline{f},\bar{f}]$, both functions satisfy
\begin{align}
    \Vert\tilde{f}_{\theta_a}\Vert_{\infty}\le1,~\Vert\hat{f}_{\theta_a}\Vert_{\infty}\le1,~\Vert\tilde{f}_{\theta_a}\Vert^2_{2}\le1, ~\Vert\hat{f}_{\theta_a}\Vert^2_{2}\le1.
\end{align}

Applying the concentration inequality for additive functionals of Markov chains \cite[Theorem~3.3]{lezaud1998chernoff}, for any $0<\epsilon_1\le1$, we have
\begin{align}
&\mathbb{P}\bigg[\Big\vert\hat{\eta}_{a,T_{j_b}^a}(\gamma)-\eta_{\theta_a}(\gamma)\Big\vert>\epsilon_1\bigg]\nn
\le &2\Vert\bs{\omega}_a\Vert_2\exp{\left(-\frac{T^a_{j_b}\epsilon_1^2\beta_{\theta_a}}{24((\frac{1}{2}\bar{f}^2-\gamma \bar{f})-(\frac{1}{2}\underline{f}^2-\gamma \underline{f}))^2}\right)},
\end{align}
where $\boldsymbol{\omega}_a$ denotes the initialization bias vector and $\beta_{\theta_a}$ is the pseudo spectral gap of arm $a$ defined in \eqref{eq:pseudo}.
Similarly, for any $0<\epsilon_2\le1$, we obtain
\begin{align}
\hspace{-0.9em}
\mathbb{P}\bigg[\Big\vert\hat{\Gamma}_{a,T_{j_b}^a}\!-\!\Gamma_{\theta_a}\Big\vert\!>\!\epsilon_2\bigg]\!\le\!2\Vert\bs{\omega}_a\Vert_2\exp{\left(\!\!-\frac{T^a_{j_b}\epsilon_2^2\beta_{\theta_a}}{24(\bar{f}^2-\underline{f}^2)^2}\!\!\right)}.
\end{align}

Using the inequality $(\frac{1}{2}M^2-\gamma M)-(\frac{1}{2}m^2-\gamma m)\le M^2-m^2$, we set $\epsilon_1=\epsilon_2=\epsilon/2$ for $\epsilon\in(0,2]$.
By the definition of the empirical estimator $\hat{g}_{\theta_a}$ in \eqref{eq:empirical age}, we obtain that
\begin{align}
\mathbb{P}\Big[\big\vert\hat{g}_{\theta_a}-&g_{\theta_a}\big\vert>\epsilon\Big]\le\mathbb{P}\bigg[\Big\vert\hat{\eta}_{a,T_{j_b}^a}(\gamma)-\eta_{\theta_a}(\gamma)\Big\vert>\epsilon_1\bigg]\nn
&\quad+\mathbb{P}\bigg[\Big\vert\hat{\Gamma}_{a,T_{j_b}^a}-\sigma^2_{\theta_a}\Big\vert>\epsilon_2\bigg]\nn
&\le4\Vert\bs{\omega}_a\Vert_2\exp{\left(-\frac{T^a_{j_b}\epsilon^2\beta_{\theta_a}}{96(\bar{f}^2-\underline{f}^2)^2}\right)}.
\end{align}

Applying a union bound over all $j_b$ sampling blocks and $K$ arms, to ensure that the concentration inequality in \eqref{eq:concentration_bound} holds with probability at least $1-\delta$, it suffices that
\begin{align}\label{eq:confidence_inequality}
4\Vert\bs{\omega}_a\Vert_2\exp{\left(-\frac{T^a_{j_b}\epsilon^2\beta_{\theta_a}}{96(\bar{f}^2-\underline{f}^2)^2}\right)}\le \frac{\delta}{j_bK}.
\end{align}
Solving \eqref{eq:confidence_inequality} for $\epsilon$ yields
\begin{align}
\epsilon\ge\sqrt{\frac{96(\bar{f}^2-\underline{f}^2)^2\log{(\frac{4\Vert\bs{\omega}_a\Vert_2j_bK}{\delta})}}{T^a_{j_b}\beta_{\theta_a}}}.
\end{align}

Using the uniform bounds $\Vert\boldsymbol{\omega}_a\Vert_2\le 1/\bar{\mu}$ and $\beta_{\theta_a}\ge\beta^{\min}$, and choosing a constant $c\ge96(\bar{f}^2-\underline{f}^2)^2/\beta^{\min}$, we simplify the expression for $\epsilon$ as
\begin{align}
\epsilon=\sqrt{\frac{c\log{(\frac{4j_bK}{\delta \bar{\mu}})}}{T^a_{j_b}}}.
\end{align}
This completes the proof.

\section{Proof of Corollary~\ref{coro:1/4}}\label{proof:coro_1/4}
Since arm $1$ is the unique optimal arm and the optimality gap of any suboptimal arm $a\neq 1$ is defined as
$\Lambda_{\theta_a}\triangleq g_{\theta_a}(\gamma^*)-g_{\theta_1}(\gamma^*)>0$ for $a\neq 1$.
Recall that we define the midpoint threshold as $\xi\triangleq \frac{1}{2}\big(g_{\theta_1}(\gamma^*)+g_{\theta_2}(\gamma^*)\big)$.
Then, for any suboptimal arm $a\neq 1$, we have
\begin{align}\label{eq:gap-midpoint}
    g_{\theta_a}(\gamma^*)-\xi&=g_{\theta_a}(\gamma^*)-\frac{g_{\theta_1}(\gamma^*)+g_{\theta_2}(\gamma^*)}{2}\nn
    &\ge g_{\theta_a}(\gamma^*)-\frac{g_{\theta_1}(\gamma^*)+g_{\theta_a}(\gamma^*)}{2}=\frac{\Lambda_{\theta_a}}{2},
\end{align}
where the inequality follows from $g_{\theta_2}(\gamma^*)\le g_{\theta_a}(\gamma^*)$ for all $a\neq 1$ as we assume $g_{\theta_a}(\gamma^*)$ is in the ascending order.

Under the good event $\mathcal{E}$ as stated in Lemma~\ref{lemma:event}, for every arm $a$ and block $b$, we first have
\begin{align}\label{eq:ci_bounds}
    \big\vert\hat{g}_a^b-g_{\theta_a}(\gamma^*)\big\vert\le \mathrm{rad}_a(b).
\end{align}
By introducing $\hat{g}_a^b\triangleq\hat{\eta}_a^b(\gamma^*)+\hat{\Gamma}_a^b$ and the definitions of the confidence bounds in \eqref{eq:confidence_interval}, we obtain that
\begin{align}
    \mathrm{LCB}_a(b)=\hat{g}_a^b-\mathrm{rad}_a(b),~~
    \mathrm{UCB}_a(b)=\hat{g}_a^b+\mathrm{rad}_a(b).
\end{align}
In particular, \eqref{eq:ci_bounds} implies $\hat{g}_a^b \ge g_{\theta_a}(\gamma^*)-\mathrm{rad}_a(b)$, and hence
\begin{align}\label{eq:lcb_lower}
    \mathrm{LCB}_a(b)=\hat{g}_a^b-\mathrm{rad}_a(b)\ge g_{\theta_a}(\gamma^*)-2\mathrm{rad}_a(b).
\end{align}
Combining \eqref{eq:gap-midpoint} and \eqref{eq:lcb_lower} yields
\begin{align}\label{eq:lucb-rad}
    \mathrm{LCB}_a(b)&\ge g_{\theta_a}(\gamma^*)-2\mathrm{rad}_a(b)\nn
    &\ge \Big(\xi+\frac{\Lambda_{\theta_a}}{2}\Big)-2\mathrm{rad}_a(b).
\end{align}

By the definition of active arms in the Age-Aware LUCB procedure, a suboptimal arm can remain active only if $\mathrm{LCB}_a(b)<\xi$.
Therefore,
it suffices to ensure $\mathrm{LCB}_a(b)\ge\xi$ to declare arm $a$ inactive.
From equation \eqref{eq:lucb-rad}, the condition $\mathrm{rad}_a(b)\le \Lambda_{\theta_a}/4$ implies $\mathrm{LCB}_a(b)\ge \xi$.
Hence, arm $a$ is no longer active, which completes the proof.

\bibliographystyle{IEEEtran}
\bibliography{reference}

@String{Computing = "Computing" }

@String{Springer = "Springer-Verlag" }

@ARTICLE{yates2021age,
  author={Yates, Roy D. and Sun, Yin and Brown, D. Richard and Kaul, Sanjit K. and Modiano, Eytan and Ulukus, Sennur},
  journal={{IEEE} J. Sel. Areas Commun.}, 
  title={Age of Information: An Introduction and Survey}, 
  year={2021},
  volume={39},
  number={5},
  pages={1183-1210}
}

@book{horn2012matrix,
  title={Matrix analysis},
  author={Horn, Roger A and Johnson, Charles R},
  year={2012},
  publisher={Cambridge University Press}
}

@article{paulin2015concentration,
  title={Concentration inequalities for Markov chains by Marton couplings and spectral methods},
  author={Paulin, Daniel},
  year={2015}
}

@article{lezaud1998chernoff,
  title={Chernoff-type bound for finite Markov chains},
  author={Lezaud, Pascal},
  journal={Annals of Applied Probability},
  pages={849--867},
  year={1998},
  publisher={JSTOR}
}

@article{dinkelbach1967nonlinear,
  title={On nonlinear fractional programming},
  author={Dinkelbach, Werner},
  journal={Management Science},
  volume={13},
  number={7},
  pages={492--498},
  year={1967},
  publisher={INFORMS}
}

@article{kaufmann2016complexity,
  title={On the complexity of best-arm identification in multi-armed bandit models},
  author={Kaufmann, Emilie and Capp{\'e}, Olivier and Garivier, Aur{\'e}lien},
  journal={J. Mach. Learn. Res.},
  volume={17},
  number={1},
  pages={1--42},
  year={2016}
}

@INPROCEEDINGS{kaul2012real,
  author={Kaul, Sanjit and Yates, Roy and Gruteser, Marco},
  booktitle={Proc. {IEEE} {INFOCOM}}, 
  title={Real-time status: How often should one update?}, 
  year={2012},
  pages={2731-2735}
}

@article{bedewy2019minimizing,
  title={Minimizing the age of information through queues},
  author={Bedewy, Ahmed M and Sun, Yin and Shroff, Ness B},
  journal={{IEEE} Trans. Inf. Theory},
  volume={65},
  number={8},
  pages={5215--5232},
  year={2019},
  publisher={IEEE}
}

@ARTICLE{bedewy2019age,
  author={Bedewy, Ahmed M. and Sun, Yin and Shroff, Ness B.},
  journal={{IEEE/ACM} Trans. Netw.}, 
  title={The Age of Information in Multihop Networks}, 
  year={2019},
  volume={27},
  number={3},
  pages={1248-1257}
}

@inproceedings{hsu2018age,
  title={Age of information: Whittle index for scheduling stochastic arrivals},
  author={Hsu, Yu-Pin},
  booktitle={Proc. {IEEE} {ISIT}},
  pages={2634--2638},
  year={2018}
}

@inproceedings{kadota2018optimizing,
  title={Optimizing age of information in wireless networks with throughput constraints},
  author={Kadota, Igor and Sinha, Abhishek and Modiano, Eytan},
  booktitle={Proc. {IEEE} {INFOCOM}},
  pages={1844--1852},
  year={2018}
}

@article{tripathi2024whittle,
  title={A whittle index approach to minimizing functions of age of information},
  author={Tripathi, Vishrant and Modiano, Eytan},
  journal={{IEEE/ACM} Trans. Netw.},
  volume={32},
  number={6},
  pages={5144--5158},
  year={2024}
}

@inproceedings{arafa2019timely,
  title={Timely cloud computing: Preemption and waiting},
  author={Arafa, Ahmed and Yates, Roy D and Poor, H Vincent},
  booktitle={Proc. {IEEE} Annual Allerton Conference on Communication, Control, and Computing},
  pages={528--535},
  year={2019}
}

@ARTICLE{yates2018age,
  author={Yates, Roy D. and Kaul, Sanjit K.},
  journal={{IEEE} Trans. Inf. Theory}, 
  title={The Age of Information: Real-Time Status Updating by Multiple Sources}, 
  year={2019},
  volume={65},
  number={3},
  pages={1807-1827}
}

@ARTICLE{talak2020age,
  author={Talak, Rajat and Modiano, Eytan H.},
  journal={{IEEE} Trans. Inf. Theory}, 
  title={Age-Delay Tradeoffs in Queueing Systems}, 
  year={2021},
  volume={67},
  number={3},
  pages={1743-1758}
}

@INPROCEEDINGS{zhou2024age,
  author={Zhou, Mengqiu and Zhang, Meng and Yang, Howard H. and Yates, Roy D.},
  booktitle={Proc. {IEEE} {INFOCOM}}, 
  title={Age-minimal CPU Scheduling}, 
  year={2024},
  pages={401-410}
}

@ARTICLE{arafa2019age,
  author={Arafa, Ahmed and Yang, Jing and Ulukus, Sennur and Poor, H. Vincent},
  journal={{IEEE} Trans. Inf. Theory}, 
  title={Age-Minimal Transmission for Energy Harvesting Sensors With Finite Batteries: Online Policies}, 
  year={2020},
  volume={66},
  number={1},
  pages={534-556}
}

@ARTICLE{sun2019sampling,
  author={Sun, Yin and Cyr, Benjamin},
  journal={J. Commun. Networks}, 
  title={Sampling for data freshness optimization: Non-linear age functions}, 
  year={2019},
  volume={21},
  number={3},
  pages={204-219}
}

@ARTICLE{sun2017update,
  author={Sun, Yin and Uysal-Biyikoglu, Elif and Yates, Roy D. and Koksal, C. Emre and Shroff, Ness B.},
  journal={{IEEE} Trans. Inf. Theory}, 
  title={Update or Wait: How to Keep Your Data Fresh}, 
  year={2017},
  volume={63},
  number={11},
  pages={7492-7508}
}

@article{whittle1988restless,
  title={Restless bandits: Activity allocation in a changing world},
  author={Whittle, Peter},
  journal={Journal of Applied Probability},
  volume={25},
  number={A},
  pages={287--298},
  year={1988},
  publisher={Cambridge University Press}
}

@ARTICLE{tekin2012online,
  author={Tekin, Cem and Liu, Mingyan},
  journal={{IEEE} Trans. Inf. Theory}, 
  title={Online Learning of Rested and Restless Bandits}, 
  year={2012},
  volume={58},
  number={8},
  pages={5588-5611}
}

@ARTICLE{Karthik2023best,
  author={Karthik, P. N. and Reddy, Kota Srinivas and Tan, Vincent Y. F.},
  journal={{IEEE} Trans. Inf. Theory}, 
  title={Best Arm Identification in Restless Markov Multi-Armed Bandits}, 
  year={2023},
  volume={69},
  number={5},
  pages={3240-3262}
}

@article{guha2010approximation,
  title={Approximation algorithms for restless bandit problems},
  author={Guha, Sudipto and Munagala, Kamesh and Shi, Peng},
  journal={J. {ACM}},
  volume={58},
  number={1},
  pages={1--50},
  year={2010},
  publisher={ACM New York, NY, USA}
}

@inproceedings{wang2020restless,
  title={Restless-UCB, an efficient and low-complexity algorithm for online restless bandits},
  author={Wang, Siwei and Huang, Longbo and Lui, John},
  booktitle={Advances in NeurIPS},
  volume={33},
  pages={11878--11889},
  year={2020}
}

@inproceedings{nakhleh2021neurwin,
  title={NeurWIN: Neural Whittle index network for restless bandits via deep RL},
  author={Nakhleh, Khaled and Ganji, Santosh and Hsieh, Ping-Chun and Hou, I and Shakkottai, Srinivas and others},
  booktitle={Advances in NeurIPS},
  volume={34},
  pages={828--839},
  year={2021}
}

@article{ortner2012regret,
  title={Regret bounds for restless markov bandits},
  author={Ortner, Ronald and Ryabko, Daniil and Auer, Peter and Munos, R{\'e}mi},
  journal={Theor. Comput. Sci.},
  volume={558},
  pages={62--76},
  year={2014}
}

@inproceedings{jiang2023online,
  title={Online restless bandits with unobserved states},
  author={Jiang, Bowen and Jiang, Bo and Li, Jian and Lin, Tao and Wang, Xinbing and Zhou, Chenghu},
  booktitle={International Conference on Machine Learning},
  pages={15041--15066},
  year={2023},
  organization={PMLR}
}

@article{liu2010indexability,
  title={Indexability of restless bandit problems and optimality of whittle index for dynamic multichannel access},
  author={Liu, Keqin and Zhao, Qing},
  journal={{IEEE} Trans. Inf. Theory},
  volume={56},
  number={11},
  pages={5547--5567},
  year={2010},
  publisher={IEEE}
}

@inproceedings{wang2023optimistic,
  title={Optimistic whittle index policy: Online learning for restless bandits},
  author={Wang, Kai and Xu, Lily and Taneja, Aparna and Tambe, Milind},
  booktitle={Proc. AAAI Conf. Artif. Intell.},
  volume={37},
  number={8},
  pages={10131--10139},
  year={2023}
}

@article{akbarzadeh2022conditions,
  title={Conditions for indexability of restless bandits and an algorithm to compute Whittle index},
  author={Akbarzadeh, Nima and Mahajan, Aditya},
  journal={Advances in Applied Probability},
  volume={54},
  number={4},
  pages={1164--1192},
  year={2022},
  publisher={Cambridge University Press}
}

@article{chen2021information,
  title={Information freshness-aware task offloading in air-ground integrated edge computing systems},
  author={Chen, Xianfu and Wu, Celimuge and Chen, Tao and Liu, Zhi and Zhang, Honggang and Bennis, Mehdi and Liu, Hang and Ji, Yusheng},
  journal={{IEEE} J. Sel. Areas Commun.},
  volume={40},
  number={1},
  pages={243--258},
  year={2021},
  publisher={IEEE}
}

@article{chen2023joint,
  title={Joint optimization of sensing and computation for status update in mobile edge computing systems},
  author={Chen, Yi and Chang, Zheng and Min, Geyong and Mao, Shiwen and H{\"a}m{\"a}l{\"a}inen, Timo},
  journal={{IEEE} Trans. Wirel. Commun.},
  volume={22},
  number={11},
  pages={8230--8243},
  year={2023},
  publisher={IEEE}
}

@article{chen2020age,
  title={Age of information aware radio resource management in vehicular networks: A proactive deep reinforcement learning perspective},
  author={Chen, Xianfu and Wu, Celimuge and Chen, Tao and Zhang, Honggang and Liu, Zhi and Zhang, Yan and Bennis, Mehdi},
  journal={{IEEE} Trans. Wirel. Commun.},
  volume={19},
  number={4},
  pages={2268--2281},
  year={2020},
  publisher={IEEE}
}

@article{chernoff1952measure,
  title={A measure of asymptotic efficiency for tests of a hypothesis based on the sum of observations},
  author={Chernoff, Herman},
  journal={The Annals of Mathematical Statistics},
  pages={493--507},
  year={1952},
  publisher={JSTOR}
}

@article{hoeffding1963probability,
  title={Probability inequalities for sums of bounded random variables},
  author={Hoeffding, Wassily},
  journal={Journal of the American Statistical Association},
  volume={58},
  number={301},
  pages={13--30},
  year={1963},
  publisher={Taylor \& Francis}
}

@book{meyn2012markov,
  title={Markov chains and stochastic stability},
  author={Meyn, Sean P and Tweedie, Richard L},
  year={2012},
  publisher={Springer Science \& Business Media}
}

@article{kadota2018scheduling,
  title={Scheduling policies for minimizing age of information in broadcast wireless networks},
  author={Kadota, Igor and Sinha, Abhishek and Uysal-Biyikoglu, Elif and Singh, Rahul and Modiano, Eytan},
  journal={{IEEE/ACM} Trans. Netw.},
  volume={26},
  number={6},
  pages={2637--2650},
  year={2018},
  publisher={IEEE}
}

@inproceedings{saurav2021minimizing,
  title={Minimizing the sum of age of information and transmission cost under stochastic arrival model},
  author={Saurav, Kumar and Vaze, Rahul},
  booktitle={Proc. {IEEE} {INFOCOM}},
  pages={1--10},
  year={2021}
}

@book{kleinrock1974queueing,
  title={Queueing systems: theory},
  author={Kleinrock, Leonard},
  volume={2},
  year={1974},
  publisher={Wiley}
}

@article{moulos2019optimal,
  title={Optimal best Markovian arm identification with fixed confidence},
  author={Moulos, Vrettos},
  journal={Advances in NeurIPS},
  volume={32},
  year={2019}
}

@article{smith1955regenerative,
  title={Regenerative stochastic processes},
  author={Smith, Walter L},
  journal={Proc. R. Soc. Lond. A Math. Phys. Sci.},
  volume={232},
  number={1188},
  pages={6--31},
  year={1955}
}

@article{lawler1988bounds,
  title={Bounds on the $L^2$ spectrum for Markov chains and Markov processes: a generalization of Cheeger’s inequality},
  author={Lawler, Gregory F and Sokal, Alan D},
  journal={Trans. Amer. Math. Soc.},
  volume={309},
  number={2},
  pages={557--580},
  year={1988}
}

@book{feller1991introduction,
  title={An introduction to probability theory and its applications, Volume 2},
  author={Feller, William},
  volume={2},
  year={1991},
  publisher={John Wiley \& Sons}
}

@article{moyne2007emergence,
  title={The emergence of industrial control networks for manufacturing control, diagnostics, and safety data},
  author={Moyne, James R and Tilbury, Dawn M},
  journal={Proc. {IEEE}},
  volume={95},
  number={1},
  pages={29--47},
  year={2007},
  publisher={IEEE}
}

@article{kakria2015real,
  title={A real-time health monitoring system for remote cardiac patients using smartphone and wearable sensors},
  author={Kakria, Priyanka and Tripathi, NK and Kitipawang, Peerapong},
  journal={Int. J. Telemed. Appl.},
  volume={},
  number={},
  pages={},
  year={2015},
  publisher={Wiley Online Library}
}

@article{liu2020computing,
  title={Computing systems for autonomous driving: State of the art and challenges},
  author={Liu, Liangkai and Lu, Sidi and Zhong, Ren and Wu, Baofu and Yao, Yongtao and Zhang, Qingyang and Shi, Weisong},
  journal={{IEEE} Internet Things J.},
  volume={8},
  number={8},
  pages={6469--6486},
  year={2020},
  publisher={IEEE}
}

@INPROCEEDINGS{9355609,
  author={Liu, Qiang and Han, Tao and Moges, Ephraim},
  booktitle={2020 IEEE 40th International Conference on Distributed Computing Systems (ICDCS)}, 
  title={EdgeSlice: Slicing Wireless Edge Computing Network with Decentralized Deep Reinforcement Learning}, 
  year={2020},
  volume={},
  number={},
  pages={234-244}
}

@ARTICLE{10173679,
  author={Shah, Syed Danial Ali and Gregory, Mark A. and Li, Shuo},
  journal={IEEE Internet of Things Journal}, 
  title={Toward Network-Slicing-Enabled Edge Computing: A Cloud-Native Approach for Slice Mobility}, 
  year={2024},
  volume={11},
  number={2},
  pages={2684-2700}
}

@ARTICLE{8016573,
  author={Mao, Yuyi and You, Changsheng and Zhang, Jun and Huang, Kaibin and Letaief, Khaled B.},
  journal={IEEE Communications Surveys \& Tutorials}, 
  title={A Survey on Mobile Edge Computing: The Communication Perspective}, 
  year={2017},
  volume={19},
  number={4},
  pages={2322-2358}
}

@ARTICLE{6824752,
  author={Andrews, Jeffrey G. and Buzzi, Stefano and Choi, Wan and Hanly, Stephen V. and Lozano, Angel and Soong, Anthony C. K. and Zhang, Jianzhong Charlie},
  journal={IEEE Journal on Selected Areas in Communications}, 
  title={What Will 5G Be?}, 
  year={2014},
  volume={32},
  number={6},
  pages={1065-1082}
}

@article{ren2019survey,
  title={A survey on end-edge-cloud orchestrated network computing paradigms: Transparent computing, mobile edge computing, fog computing, and cloudlet},
  author={Ren, Ju and Zhang, Deyu and He, Shiwen and Zhang, Yaoxue and Li, Tao},
  journal={ACM Computing Surveys (CSUR)},
  volume={52},
  number={6},
  pages={1--36},
  year={2019},
  publisher={ACM New York, NY, USA}
}

\end{document}